\documentclass[11pt]{article}
\usepackage{bbding}
\fussy
\usepackage{pdfsync}
\usepackage{framed, xcolor}
\usepackage{tocvsec2}
\usepackage{datetime}
\usepackage{pifont}
\usepackage{pdflscape}
\usepackage{subfigure}          
\usepackage{colortbl}
\usepackage{booktabs}
\usepackage{pdfsync}
\usepackage[font=scriptsize,bf]{caption}
\usepackage{tikz,subfigure}
\usepackage[active]{srcltx}
\usepackage[margin=1in]{geometry}
\usepackage{algorithm}
\usepackage{algorithmic}
\usepackage{epsfig,amssymb,amsfonts,amsmath,amsthm}
\usepackage{multirow}
\usepackage[numbers,sort&compress,sectionbib]{natbib}
\usepackage{amsfonts}
\usepackage{xspace}
\usepackage{tabularx}
\usepackage{pstricks}
\usepackage{setspace}
\usepackage{xcolor}

\usepackage{rotating}
\usepackage{minitoc}

\usepackage{tikz}
\usepackage{enumerate}

\usepackage{hyperref}
\usetikzlibrary{chains,fit,shapes,arrows}
\usetikzlibrary{shapes.arrows}

\hypersetup{colorlinks=true, linkcolor=red}

\newcommand{\phiw}{\phi_{w}}
\newcommand{\epsI}{\eps_{I}}

\newcommand{\SDDM}{\mathrm{SDDM}}
\newcommand{\SDD}{\mathrm{SDD}}

\newcommand{\SPSD}{\mathrm{SPSD}}
\newcommand{\MDBD}{\mathrm{MDBD}}

\newcommand{\GL}{\mathcal{T}}

\newcommand{\BNTap}{\mathcal{B}_{N,T}(\alpha,p)}
\newcommand{\GLB}{\mathrm{P}_{\mathcal{B}}}

\newcommand{\mGg}{\mathcal{G}_{\gamma}}

\newcommand{\mKLC}{\mathrm{\mathbf{mKLC}}}
\newcommand{\pKLC}{\mathrm{\mathbf{pKLC}}}
\newcommand{\iSS}{\mathrm{\mathbf{SS}}}

\newcommand{\mPS}{\mathrm{\mathbf{mPS}}}
\newcommand{\mSS}{\mathrm{\mathbf{mSS}}}
\newcommand{\fSS}{\mathrm{\mathbf{fSS}}}

\newcommand{\InitSS}{\mathrm{\mathbf{InitSS}}}
\newcommand{\SqrSS}{\mathrm{\mathbf{SqrSS}}}
\newcommand{\pSqrSS}{\mathrm{\mathbf{pSqrSS}}}
\newcommand{\IndSS}{\mathrm{\mathbf{IndSS}}}
\newcommand{\pIndSS}{\mathrm{\mathbf{pIndSS}}}

\newcommand{\PwrSS}{\mathrm{\mathbf{PwrSS}}}
\newcommand{\SSMDBD}{\mathrm{\mathbf{SS\_MDBD}}}
\newcommand{\pSSMDBD}{\mathrm{\mathbf{pSS\_MDBD}}}

\newcommand{\AppDscrPDF}{\mathrm{\mathbf{AppDscrPDF}}}

\newcommand{\pPS}{\mathbf{pPS}}
\newcommand{\PS}{\mathbf{PS}}
\newcommand{\STOVP}{\mathbf{STOVP}}

\newcommand{\wD}{\widetilde{D}}
\newcommand{\wA}{\widetilde{A}}
\newcommand{\wM}{\widetilde{M}}
\newcommand{\wO}{\widetilde{O}}

\newcommand{\hM}{\widehat{M}}
\newcommand{\hA}{\widehat{A}}
\newcommand{\hB}{\widehat{B}}

\newcommand{\hO}{\widehat{O}}
\newcommand{\hw}{\widehat{w}}

\newcommand{\Vp}{\mathbf{V}(p)}

\newcommand{\Di}{D^{-1}}
\newcommand{\Dhp}{D^{1/2}}
\newcommand{\Dhm}{D^{-1/2}}

\newcommand{\mD}{\mathbf{D}}
\newcommand{\mL}{\mathbf{L}}
\newcommand{\mB}{\mathbf{B}}

\newcommand{\mBNp}{\mathbf{B}_{N}(p)}

\newcommand{\diag}{\mathrm{diag}}
\newcommand{\prm}{\prime}

\newcommand{\remove}[1]{}

\newcommand{\N}{\mathbb{N}}
\newcommand{\R}{\mathbb{R}}

\newcommand{\rot}{\mathrm{T}}

\newcommand{\poly}{\operatorname{poly}}

\newcommand{\eps}{\epsilon}

\renewcommand{\leq}{\leqslant}
\renewcommand{\geq}{\geqslant}

\newcommand{\lemref}[1]{Lemma~\ref{lem:#1}}

\renewcommand{\eps}{\varepsilon}

\newcommand{\mylemma}[2]{\begin{lem}\label{lem:#1}#2\end{lem}}

\newtheorem{problem}{Problem}
\newtheorem{thm}{Theorem}  
\newtheorem{fact}[thm]{Fact}
\newtheorem{lem}[thm]{Lemma}

\newtheorem{clm}[thm]{Claim}
\newtheorem{cor}[thm]{Corollary}

\newtheorem{rem}[thm]{Remark}

\numberwithin{thm}{section}

\newcommand{\mat}[1]{\boldsymbol{\mathbf{#1}}}

\title{An Efficient Parallel Algorithm for Spectral Sparsification of Laplacian and SDDM Matrix Polynomials}

\author{
Gorav Jindal \qquad Pavel Kolev\footnote{This work has been funded by the Cluster of Excellence ``Multimodal Computing and Interaction" within the Excellence Initiative of the German Federal Government.}\\
Max-Planck-Institut f\"{u}r Informatik, Saarbr\"{u}cken, Germany\\
\{gjindal,pkolev\}@mpi-inf.mpg.de
}

\date{}

\begin{document}

\maketitle

\begin{abstract}
A mixture of discrete Binomial distributions ($\MDBD$), denoted by $\BNTap$, is a set of pairs $\{(B(p_{i},N),\alpha_{i})\}_{i=1}^{T}$, where $B(\cdot,\cdot)$ denotes the Binomial distribution, all $p_{i}\in(0,1)$ are distinct, $\sum_{i=1}^{T}\alpha_{i}\leq1$ and all $\alpha_{i}\in(0,1)$. A vector $\gamma$ is induced by $\MDBD$ if $\gamma_{i}=\sum_{j=1}^{T}\alpha_{j}\cdot B_{N,i}(p_{j})$ for all $i\in[0:N]$, where $B_{N,i}(p)={N \choose i}p^{i}(1-p)^{N-i}$.

We prove for ``large'' class $\mathcal{C}$ of continuous probability density functions (p.d.f.), that for every $w\in\mathcal{C}$ there exists $\MDBD$ with $T\geq N\sqrt{\phiw/\delta}$ that $\delta$-approximates a \emph{discretized} p.d.f. $\widehat{w}(i/N)\triangleq w(i/N)/[\sum_{\ell=0}^{N}w(\ell/N)]$ for all $i\in[3:N-3]$, where $\phiw\geq\max_{x\in[0,1]}|w(x)|$. Moreover, we propose an efficient parallel algorithm that on input p.d.f. $w\in\mathcal{C}$ and parameter $\delta>0$, outputs $\MDBD$ that induces a vector $\gamma$ which $\delta$-approximates $\widehat{w}$. Also, we give an efficient parallel algorithm that on input a discretized p.d.f. $\widehat{w}$ induced by $\MDBD$ with $T=N+1$ and the corresponding vector $p$, outputs exactly the coefficients $\alpha$.

Cheng et al.~\cite{CCLPT15} proposed the first sequential algorithm that on input a discretized p.d.f. $\beta$, $B=D-M$ that is either Laplacian or $\SDDM$ matrix and parameter $\eps\in(0,1)$, outputs in time $\hO(\eps^{-2} m N^2)$\footnote{$\hO(\cdot)$ notation hides $\poly(\log n,\log N)$ factors.} a spectral sparsifier of a matrix-polynomial $D-\hM_{N} \approx_{\eps} D-D\sum_{i=0}^{N}\beta_{i}(\Di M)^i$. However, given $\MDBD$ $\BNTap$ that induces a discretized p.d.f. $\gamma$, to apply the algorithm in~\cite{CCLPT15} one has to explicitly precompute in $O(NT)$ time the vector $\gamma$. Instead, we give two algorithms (sequential and parallel) that bypass this explicit precomputation.

We propose a faster sequential algorithm that on input $\MDBD$ $\BNTap$ with $N=2^k$ for $k\in\N_+$ outputs in $\hO(\eps^{-2}m + \eps^{-4}nT)$ time the desired spectral sparsifier. Moreover, our algorithm is parallelizable and runs in $\hO(\eps^{-2}m + \eps^{-4}nT)$ work and $O(\log N\cdot\poly(\log n)+\log T)$ depth. Our main algorithmic contribution is to propose the first efficient parallel algorithm that on input continuous p.d.f. $w\in\mathcal{C}$, matrix $B=D-M$ as above, outputs a spectral sparsifier of matrix-polynomial whose coefficients approximate component-wise the discretized p.d.f. $\hw$.

Our results yield the first efficient and parallel algorithm that runs in nearly linear work and poly-logarithmic depth and analyzes the long term behaviour of Markov chains in non-trivial settings. In addition, we strengthen the Spielman and Peng's~\cite{PS14} parallel $\SDD$ solver by introducing a simple parallel preprocessing step.
\end{abstract}

\thispagestyle{empty}

\setcounter{page}{0}

\newpage
\tableofcontents
\newpage

\section{Introduction\label{sec:Intro}}

In their seminal work Spielman and Teng~\cite{ST14} introduced the
notion of spectral sparsifiers and proposed the first nearly linear
time algorithm for spectral sparsification. In consecutive work, Spielman
and Srivastava~\cite{SS08} proved that spectral sparsifiers with
$O(\eps^{-2}n\log n)$ edges exist and can be computed
in $\wO(m\log^{c}n\cdot\log (w_{max}/w_{min}))$\footnote{The $\wO(\cdot)$ notation hides $O(\poly(\log\log n))$ factors.} time for any undirected graph $G=(V,E,w)$. The computational bottleneck of their algorithm is to approximate
the solutions of logarithmically many $\SDD$\footnote{$\SDD$ is the class of symmetric and diagonally dominant matrices.} systems.

Recently, Koutis, Miller and Peng~\cite{KMP11} developed
an improved solver for $\SDD$ systems that works in $\wO(m\log n\cdot\log (1/\eps))$
time. In a survey result~\cite[Theorem 3]{KL13} Kelner and Levin showed that in $\wO(m\log^{2}n)$ time all effective resistances can be approximated up to a constant factor. This yields a $(1\pm\eps)$-spectral sparsifier with only a constant factor blow-up of non-zero edges $O(\eps^{-2}n\log n)$. Although there are faster by a $\mathrm{poly}\log$-factor sparsification algorithms~\cite{KLP12} they output spectral sparsifiers with $\mathrm{poly}\log$-factor more edges.

Spielman and Peng~\cite{PS14} introduced the notion of
\emph{sparse approximate inverse chain} of $\SDDM$\footnote{$\SDDM$ is the class of positive definite $\SDD$ matrices with non-positive off-diagonal entries.} matrices. They proposed the first parallel algorithm that finds such chains and runs in work $\wO(m\log^{3}n\cdot\log^{2}\kappa)$ and depth $O(\log^{c}n\cdot\log\kappa)$, where $\kappa$
is the condition number of the $\SDDM$ matrix with $m$ non-zero entries and dimension $n$. Furthermore, they showed that in $\wO(\eps^{-2}m\log^{3}n)$
 time a spectral sparsifier $\wD-\wA\approx_{\eps}D-A\Di A$
can be computed with $nnz(\wA)\leq O(\eps^{-2}n\log n)$.
In a follow up work, Cheng et al.~\cite{CCLPT14} designed an algorithm that computes a sparse approximate generalized chain $\widetilde{C}$ such that
$\widetilde{C}\widetilde{C}^{\rot}\approx_{\eps}M^{p}$ for any $\SDDM$ matrix $M$ and $|p|\leq1$. The chain $\widetilde{C}$ is constructed iteratively and it involves a normalization step that produces a sparsifier $D-\widetilde{M_{i+1}}\approx_{\eps}D-\widetilde{M_{i}}\Di \widetilde{M_{i}}$ that is expressed in terms of the original diagonal matrix $D$, for all iterations $i$.

Sinclair and Jerrum~\cite{SJ89} analyzed Markov chains with transition matrices $W=[I+\Di A]/2$, corresponding to lazy random walks. They proved that these walks converge fast to stationary distribution, defined by $\pi_u=d_u/(\sum_{u\in V}d_u)$, after $O(\phi_{G}^{-2}\log(\min_{u\in V}\pi_u^{-1}))$\footnote{Graph conductance $\phi_{G}\triangleq\min_{\mu(S)\leq \mu(V)/2}\phi(S)$, where $\phi(S)\triangleq|w(S,\overline{S})|/\mu(S)$ and $\mu(S)\triangleq\sum_{u\in S}d_u$. } steps. Andersen et al.~\cite{ACL06} gave an efficient local clustering algorithm that relies on a lazy variation of PageRank, the transition matrix of which is defined by $\sum_{t=0}^{\infty}\alpha(1-\alpha)^{t}W^t$, where $\alpha>0$ is a parameter. Their local algorithm uses a truncated (finite summation) version of the preceding transition matrix.

Recently, Cheng et al.~\cite{CCLPT15} initiated the study of computing spectral sparsifiers $D-\hA\approx_{\eps}D-D\sum_{i=1}^{N}\xi_{i}(\Di A)^{i}$ of random walk Laplacian matrix polynomials, where $\xi$ is a probability distribution over $[1:N]$, $D-A$ is a Laplacian matrix and $\sum_{i=1}^{N}\xi_{i}(\Di A)^{i}$ is a random walk transition matrix. These matrix polynomials capture the long term behaviour of Markov chains. Moreover, a sparsifier of a matrix polynomial yields a multiplicative approximation of the expected generalized ``escaping probability''~\cite{OT12,KS14} of random walks. Cheng et al.~\cite{CCLPT15} gave the first sequential algorithm that computes a spectral sparsifier of a random walk Laplacian matrix polynomial and runs in time $\wO(\eps^{-2}\cdot mN^{2}\cdot\log^{c_{1}}n\cdot\log^{c_{2}}N)$ for some small constants $c_{1},c_{2}$.

\section{Our Results}

The lazy random walk length $N$ in the regime of interest in~\cite{SJ89,ACL06,OT12} is of order $N=\Theta(\poly(n))$. The quadratic runtime dependance on $N$ makes the algorithm in~\cite{CCLPT15} prohibitively expensive for analysing the long term behaviour of Markov chains. In this paper, we overcome this issue for ``large'' class of probability distributions $\gamma$ over $[0:N]$ that are induced by mixture of discrete Binomial distributions ($\MDBD$) with $N=2^k$ for $k\in\N_+$. Our results are summarized as follows.

In Subsection~\ref{subsec:RPMDBD}, we analyze the representational power of $\MDBD$. In Subsection~\ref{subsec:SSMP}, we give a sequential and a parallel algorithm for computing a spectral sparsifier of matrix polynomials induced by $\MDBD$. In Subsection~\ref{subsec:Apps}, we propose the first parallel algorithm that runs in nearly linear work and poly-logarithmic depth and analyzes the long term behaviour of Markov chains in non-trivial settings. In Subsection~\ref{subsec:FSSMS}, we strengthen the Spielman and Peng's~\cite{PS14} parallel $\SDD$ solver.

\subsection{Representational Power of $\MDBD$}\label{subsec:RPMDBD}

Let $B(p,N)$ be Binomial distribution for some parameters $p\in(0,1)$ and $N\in\N_+$. $\MDBD$ is a set of pairs $\{(B(p_{i},N),\alpha_{i})\}_{i=1}^{T}$, denoted by $\BNTap$, that satisfies the following two conditions:

1. (distinctness) $\, p_{i}\in(0,1)$ and $p_{i}\neq p_{j}$ for all $i\neq j\in[1:T]$;

2. (positive linear combination) $\, \sum_{i=1}^{T}\alpha_{i}\leq1$ and $\alpha_{i}\in(0,1)$ for all $i\in[1:T]$.

We prove in Section~\ref{sec:APMDBD} that for every function $w$ in a ``large'' class of continuous p.d.f., there exists $\MDBD$ that induces a component-wise approximation of $w$.

\begin{thm}[$\MDBD$ Yields a Component-Wise Approximation]\label{thmMDBD}
Let $w(x)$ be a four times differentiable p.d.f., $\eps_{I}>0$ a parameter and $I=[0,1]$ an interval. Suppose there is an integer $N_{0}\in\N$ and reals $\mu\in(0,1)$ and $\phiw\geq1$ such that:\\
1) $\max_{x\in I}|w^{\prm\prm}(x)|\leq 2\phiw \cdot N_{0}^{2}$,$\,\,\,\,\,\,\,$ 2) $\max_{x\in I}|w^{\prm}(x)|\leq\frac{1}{2}\phiw\cdot N_{0}$,$\,\,\,\,\,\,\,\,\,\,\,\,\,\,$3) $\max_{x\in I}|w(x)|\leq \phiw$,\\
4) $\max_{x\in I}|b_{2}(x)|\leq\frac{1}{2}\mu\cdot N_{0}^{2}$ ,$\quad\,\,\,\,\,\,$5) $\max_{x\in I}|b_{1}(x)|\leq\frac{1}{2}\mu\cdot N_{0}$,\\
where the functions $b_1,b_2$ are defined by $b_{1}(x)=\frac{1}{w(x)}[-w(x)+(1-2x)w^{\prm}(x)+\frac{1}{2}x(1-x)w^{\prm\prm}(x)]$
and $b_{2}(x)=\frac{1}{w(x)}[w(x)-3(1-2x)w^{\prm}(x)+(1-6x+6x^{2})w^{\prm\prm}(x)+\frac{5}{6}x(1-x)(1-2x)w^{\prm\prm\prm}(x)+\frac{1}{8}x^{2}(1-x)^{2}w^{\prm\mathrm{v}}(x)].$
Then for every $N\geq N_{0}$, any $T\geq\Omega(N\sqrt{\phiw/\eps_{I}})$ and all $i\in[3:N-3]$ there is $\eta_{i}\in[-\mu,\mu]$ such that
\begin{equation}\label{eq:fin_Approx}
\left|(1+\eta_{i})\frac{w(i/N)}{N} - \sum_{j=1}^{T}F_{i}(j/[T+1])\right| \leq \frac{\eps_{I}}{N},
\end{equation}
where $F_{i}(x)\triangleq (w(x)/[T+1])\cdot B_{N,i}(x)$ and $B_{N,i}(x)\triangleq{N \choose i}x^{i}(1-x)^{N-i}$.
\end{thm}

For every function $w$ that satisfies the hypothesis in Theorem~\ref{thmMDBD} we associate $\MDBD$ $\BNTap$ that is defined by $p_j=j/(T+1)$ and $\alpha_j=w(p_j)/(T+1)$ for all $j\in[1:T]$. Moreover, in Section~\ref{sec:approxDiscrPDFs} we give an efficient parallel algorithm that on input a continuous p.d.f. $w$ (satisfying the conditions in Theorem~\ref{thm_AppDscrPDF}) and integer $N\in\N_+$, outputs $\MDBD$ that induces a discretized p.d.f. which approximates component-wise a desired discretized p.d.f. $\hw$.

\begin{thm}[An Efficient Parallel Algorithm for Finding $\MDBD$]\label{thm_AppDscrPDF}
Let $w(x)=C\cdot f(x)$ be a p.d.f. that satisfies the hypothesis of Theorem~\ref{thmMDBD}. Suppose also $a)\;0\leq f(x)\leq1$,\,\,\,\,$b)\;\frac{1}{2}[f(0)+f(1)] \geq \Omega(1)$,\,\,\,\,$c)\;1\leq C\leq o(N)$,\,\,\,\,$ and $\,\,\,\,$d)\;\int_{0}^{1}\left|f^{(2)}(x)\right|dx\leq o(N)$. Then there is a parallel algorithm $\AppDscrPDF$ that on input $w(x)$ as above, integer $N\in\N_+$ and parameter $\eps_{I}>0$, outputs in $O(N\sqrt{\phiw/\epsI})$ work and $O(\log(N\sqrt{\phiw/\epsI}))$ depth $\MDBD$ $\BNTap$ that induces a discretized probability distribution $\gamma/[1-\delta_w]$ over $[0:N]$ such that for all $i\in[3:N-3]$
\[
\frac{\gamma_{i}}{1-\delta_{w}}\in \left[(1+2\eta_{i})\cdot\hw(i/N) \pm \frac{2\eps_{I}}{S_{N+1,N}}\right],
\]
where the target discretized p.d.f. is defined by $\hw(i/N)\triangleq w(i/N)/S_{N+1,N}$ for all $i\in[0:N]$, $\delta_w\triangleq 1-\frac{S_{T,T+1}/(T+1)}{S_{N+1,N}/N}$, $S_{N+1,N}\triangleq\sum_{j=0}^{N}w(j/N)$ and $S_{T,T+1}\triangleq\sum_{k=1}^{T}w(j/[T+1])$. Moreover, it holds that $\delta_w\in[0,o(1)]$.
\end{thm}

In Appendix~\ref{appsec:AGPD}, we illustrate the representational power of $\MDBD$ by applying Theorem~\ref{thm_AppDscrPDF} for two canonical continuous p.d.f.: the Uniform distribution and the Exponential Families.

\paragraph*{Exact Recovery}
Interestingly, in case when a discretized p.d.f. $\hw$ is induced by $\MDBD$ $\BNTap$ with exactly $T=N+1$ distinct Binomial distributions, we give in Section~\ref{sec:StBVls} an efficient parallel algorithm that on input vectors $p$ and $\hw$, outputs the vector $\alpha$ in $O(N\log^{2}N)$ work and $O(\log^{c}N)$ depth for some constant $c\in\N_+$.

\begin{thm}[Canonical Instances Admit Exact Recovery]\label{thm_my_Inv_Bnp}
Suppose $p\in(0,1)^{N+1}$ is a vector
such that $0<p_{i}\neq p_{j}<1$ for every $i\neq j$ and $\hw\in(0,1)^{N+1}$ is a discretized p.d.f. that is induced by $\MDBD$ $\mathcal{B}_{N,N+1}(\alpha,p)$
that satisfies $\hw(i)=\sum_{j=1}^{N+1}\alpha_{j}\cdot B_{N,i}(p_{j})$
for every $i\in[0:N]$. Then there is a parallel algorithm that
on input the vectors $p$ and $\hw$, outputs the vector $\alpha\in(0,1)^{N+1}$ in $O(N\log^{2}N)$ work and $O(\log^{c}n)$ depth, for some constant $c\in\N_+$.
\end{thm}

\subsection{Spectral Sparsification of Matrix Polynomials induced by $\MDBD$}\label{subsec:SSMP}

A matrix $B$ is $\GL$-matrix if it is either Laplacian or $\SDDM$ matrix. To highlight that an algorithm $\mathcal{A}$ preserves the matrix type, we write that the algorithm $\mathcal{A}$ on input a $\GL$-matrix $B$ outputs a matrix $B^\prm$ that is also $\GL$-matrix.

Moreover, we say that a matrix $X$ is a spectral sparsifier of a matrix $Y$ if it
satisfies $(1-\eps)Y\preceq X\preceq(1+\eps)Y$, for short $X\approx_{\eps}Y$,
where the partial relation $X\succeq0$ stands for $X$ is symmetric
positive semi-definite (SPSD) matrix.

We denote by $nnz(A)$ or $m_A$ the number of non-zero entries of matrix $A$. When we write ``$B=D-M$ is $\GL$-matrix'' we assume that $D$ is positive diagonal matrix and $B\in\R^{n\times n}$. All algorithms presented in this paper output spectral sparsifiers \emph{with high probability}.

\paragraph{Sequential Algorithms}
Cheng et al.~\cite[Theorem 1.5]{arxivCCLPT15} gave an algorithm that on input a Laplacian matrix $L=D-A$, even integer $N\in\N_+$ and parameter $\eps>0$, outputs in time $O(\eps^{-2}m_L\log^{3}n\cdot\log^2{N})$ a spectral sparsifier $D-\hA\approx_{\eps}D-D(\Di A)^N$ of a matrix-monomial such that $nnz(\hA)\leq O(\eps^{-2}n\log n)$. In Section~\ref{sec:CIA}, we give for $N=2^k$ and $k\in\N_+$ a $O(\log^{2}N)$-factor faster algorithm that computes a spectral sparsifier of $\GL$-matrix monomials. Furthermore, for any $\GL$-matrix $D-M$ such that $M$ is $\SPSD$ matrix, we prove that the initial sparsification step dominates the algorithm's runtime.

\begin{thm}[Power Method for Monomials]\label{thmMyGLN}
There is an algorithm $\PwrSS$ that on input $\GL$-matrix $B=D-M$, $N=2^{k}$ for $k\in\N_+$ and $\eps\in(0,1)$, outputs a spectral sparsifier
$D-\hM_{N}\approx_{\eps}D-D(\Di M)^{N}$ that is $\GL$-matrix with $nnz(\hM_{N})\leq O(\eps^{-2}n\log n)$.
The algorithm runs in time
\[
\begin{cases}
\wO(m_{B}\log^{2}n+\eps^{-4}\cdot n\log^{4}n\cdot\log^{5}N) & \text{, if \ensuremath{M} is \ensuremath{\SPSD} matrix};\\
\wO(\eps^{-2}m_{B}\log^{3}n + \eps^{-4}\cdot n\log^{4}n\cdot\log^{5}N) & \text{, otherwise.}
\end{cases}
\]
\end{thm}

Using Theorem~\ref{thmMyGLN}, we give in Section~\ref{sec:SSBGLPM} an algorithm
that runs by $\Theta(N^2)$-factor faster than~\cite[Theorem 2]{CCLPT15} and computes a spectral sparsifier of a single Binomial $\GL$-matrix polynomials of the form $D-D\sum_{i=0}^{N}B_{N,i}(p)\cdot(\Di M)^{i}=D-DW_{p}^{N}$, where $W_{p}=(1-p)I+p\Di M$ and $p\in(0,1)$.

\begin{thm}[Single Binomial Matrix Polynomials]\label{thm_LazySS}
There is an algorithm $\mathrm{\mathbf{LazySS}}$
that on input $\GL$-matrix $B=D-M$, number $N=2^{k}$
for $k\in\N_+$, and parameters $\eps,p\in(0,1)$, outputs a spectral sparsifier
\[
D-\hM_{p,N} \approx_{\eps} D-DW_{p}^{N}=D - D\sum_{i=0}^{N} B_{N,i}(p) \cdot(\Di M)^{i}
\]
that is $\GL$-matrix with at most $O(\eps^{-2}n\log n)$ non-zero entries. The algorithm $\mathrm{\mathbf{LazySS}}$ runs in time
\[
\begin{cases}
\wO(m_{B}\log^{2}n+\eps^{-4}\cdot n\log^{4}n\cdot\log^{5}N) & \text{, if \ensuremath{p\in(0,1/2]}};\\
\wO(\eps^{-2}m_{B}\log^{3}n\cdot\log^{2}N+\eps^{-4}\cdot n\log^{4}n\cdot\log^{5}N) & \text{, otherwise.}
\end{cases}
\]
\end{thm}

In Section~\ref{sec:SSGLBM}, we give our main sequential algorithm that builds upon Theorem~\ref{thm_LazySS} and computes a spectral sparsifier of $\GL$-matrix polynomials induced by $\MDBD$.

\begin{thm}[Mixture of Binomial Matrix Polynomials]\label{thm_SS_MGL}
There is an algorithm $\SSMDBD$ that on input $\GL$-matrix $B=D-M$, integer $N=2^{k}$, $k\in\N_+$, $\MDBD$ $\BNTap$ with $\delta=1-\sum_{i=1}^{T}\alpha_i$ and parameter $\eps\in(0,1)$, outputs a spectral sparsifier $D-\hM \approx_{\eps}D-D\sum_{i=0}^{N}(\gamma_{i}/[1-\delta])\cdot(\Di M)^{i}$ that is $\GL$-matrix with $O(\eps^{-2}n\log n)$ non-zero entries, where $\gamma$ is a discretized p.d.f. that satisfies $\gamma_{i}=\sum_{j=1}^{T}\alpha_{j}\cdot B_{N,i}(p_{j})$ for all $i$. The algorithm runs in time
\[
\begin{cases}
\wO(m_B\log^{2}n + \eps^{-4}\cdot nT\cdot\log^{4}n\cdot\log^{5}N) & \text{, if $M$ is $\SPSD$ matrix};\\
\wO(\eps^{-2}m_B\log^{3}n + \eps^{-4}\cdot nT\cdot\log^{4}n\cdot\log^{5}N) & \text{, otherwise.}
\end{cases}
\]
\end{thm}

We motivate now the first conclusion of Theorem~\ref{thm_SS_MGL}. When $B=D-D W_p$ is a Laplacian matrix (of a graph with self-loops) associated with a Markov Chain with transition matrix $W_p$ that corresponds to $p$-lazy random walk process, it holds for $p\in(0,1/2]$ (c.f. \lemref{lemDDWWp}) that $D W_p$ is $\SPSD$ matrix.

Given $\MDBD$ $\BNTap$ that induces a vector $\gamma$, the algorithm in~\cite[Theorem 2]{CCLPT15} outputs a spectral sparsifier of the corresponding $\GL$-matrix polynomial in time $\hO(\eps^{-2} m N^2+NT)$\footnote{$\hO(\cdot)$ notation hides $\poly(\log n,\log N)$ factors.}, where the term $O(NT)$ accounts for computing the vector $\gamma$. In comparison, our improved algorithm $\SSMDBD$ runs in time $\hO(\eps^{-2}m+\eps^{-4}nT)$ for any $\MDBD$ with $N=2^k$ for $k\in\N_+$.

\paragraph{Parallel Algorithms}

Building upon the seminal works of Spielman and Teng~\cite{ST11}, Orecchia and Vishnoi~\cite{OV11} and Spielman and Peng~\cite{PS14}, we prove in Section~\ref{sec:effPrlAlg} that algorithm $\SSMDBD$ can be efficiently parallelized. To the best of our knowledge, this is the first efficient and parallel algorithm that sparsifies $\GL$-matrix polynomials induced by $\MDBD$ with $N=2^k$ for $k\in\N_+$.

\begin{thm}[Efficient Parallel Spectral Sparsification of Matrix Polynomial induced by $\MDBD$]\label{thm_prl_SS_MGL}
There is a parallel algorithm $\pSSMDBD$ that on
input as in Theorem~\ref{thm_SS_MGL}, outputs a spectral sparsifier $D-\hM \approx_{\eps}D-D\sum_{i=0}^{N}(\gamma_{i}/[1-\delta])\cdot(\Di M)^{i}$ that is $\GL$-matrix with $nnz(\hM)\leq O(\eps^{-2}n\log^{c} n)$ for some constant $c$, where $\gamma$ is a discretized p.d.f. such that $\gamma_{i}=\sum_{j=1}^{T}\alpha_{j}\cdot B_{N,i}(p_{j})$ for all $i$. The algorithm runs in work $\wO(\eps^{-2}m_B\log^{c_1+1}n\cdot\log^{2}N + \eps^{-4}\cdot nT\cdot\log^{c+c_1+1}n\cdot\log^{5}N)$ and depth $O(\log^{c_2}n\cdot\log N+\log T)$ for constants $c_1,c_2\in\N$.
\end{thm}

By combining Theorem~\ref{thm_AppDscrPDF} and Theorem~\ref{thm_prl_SS_MGL}, we develop an efficient parallel algorithm that outputs a spectral sparsifier of a $\GL$-matrix polynomial whose coefficients approximate component-wise a target discretized p.d.f. $\hw$.

\begin{cor}[Approximating Target Transition Matrices]
There is a parallel algorithm that takes as input a continuous p.d.f. $w$ satisfying the conditions of Theorem~\ref{thm_AppDscrPDF}, $\GL$-matrix $B=D-M$ and parameters $\eps,\epsI\in(0,1)$, and it outputs in $\hO(\eps^{-2}m_B + \eps^{-4}nN\sqrt{\phiw/\epsI})$ work and $O(\log^{c_2}n\cdot\log N+ \log(N\sqrt{\phiw/\epsI}))$ depth a spectral sparsifier $D-\hM \approx_{\eps}D-D\sum_{i=0}^{N}(\gamma_{i}/[1-\delta_{w}])\cdot(\Di M)^{i}$ that is $\GL$-matrix with $nnz(\hM)\leq O(\eps^{-2}n\log^{c} n)$ such that for all $i\in[3:N-3]$ it holds
$\gamma_{i}/[1-\delta_{w}]\in [(1+2\eta_{i})\cdot\hw(i/N) \pm 2\eps_{I}/S_{N+1,N}]$.
\end{cor}

\subsection{Analyzing the Long Term Behaviour of Markov Chains}\label{subsec:Apps}

For many finite Markov chains~\cite{SJ89,ACL06,OT12,KS14} there exists $N^\prm\in\N_+$ such that for every $N\geq N^\prm$ certain phenomenon occurs with high probability - (local) mixing time, truncated PageRank, etc. Therefore, to analyze the long term behaviour of a finite Markov chain, it suffices to select the smallest $N=2^k$ for $k\in\N_+$ that is larger or equal to $N^\prm$.

\begin{cor}[Capturing The Long Term Behaviour of Markov Chains]\label{cor:denseTmtx}
Suppose $L=D-A$ is dense Laplacian matrix with $m_{L}=\Theta(n^{2})$, $\eps\in(0,1)$, $\BNTap$ is $\MDBD$ such that $\sum_{i=1}^{T}\alpha_{i}=1$, the degree $N=2^k\leq O(n)$ for $k\in\N_+$ and the number of Binomials $T\leq O(n)$. Then algorithm $\pSSMDBD$ outputs in work $\wO(\eps^{-4}\cdot m_L\cdot\log^{c+c_1+6}n)$ and depth $O(\log^{c_2+1}n)$ a spectral sparsifier $D-\hA \approx_{\eps}D-D\sum_{i=0}^{N}\gamma_{i}(\Di A)^{i}$ that is Laplacian matrix with $O(\eps^{-2}n\log^{c_2} n)$ non-zero entries for some constants $c,c_1,c_2\in\N$ and $\gamma$ is a probability distribution induced by the $\MDBD$ $\BNTap$.
\end{cor}

\paragraph*{Multiplicative Approximation of Generalized Escaping Probability}

Consider a Markov chain with transition matrix $\mGg=\sum_{i=0}^{N}\gamma_{i}(\Di A)^i$ that corresponds to a generalized random walk process of length $N$. Perform a random walk of length $N$ induced by $\mGg$ that starts at vertex $v\in V$. Then for any subset $S\subset V$, the corresponding \emph{generalized escaping probability} is defined by $\mathrm{gEsc}(v,S,\mGg)=\mat{1}_{v}^{\rot}\mGg\mat{1}_{\overline{S}}$, where we denote by $\mat{1}_{T}$ the characteristic vector of a subset $T\subset V$.

We define the volume of $S$ by $\mu(S)=\sum_{u\in S}d_u$ and let $\pi_S$ be a probability distribution over $V$ defined by $\pi_{S}(u)=d_u/\mu(S)$ if $u\in S$ and $\pi_{S}(u)=0$ otherwise. The \emph{expected} generalized escaping probability (E.G.E.P.) with respect to $\pi_S$ is defined by
\begin{equation}\label{eq:ExpGEP}
\mathbb{E}_{v\sim\pi_S}[\mathrm{gEsc}(v,S,\mGg)]=\pi_{S}^{\rot}\mGg\mat{1}_{\overline{S}}.
\end{equation}

We show in Appendix~\ref{appsec:GenEscProb} that a spectral sparsifier of a random walk Laplacian matrix polynomial, yields a multiplicative approximation of E.G.E.P. for all subsets $S\subset V$.

\begin{lem}[Multiplicative Approximation of E.G.E.P.]\label{lem_mulApproxEGEP}
For any spectral sparsifier $D-\hA_{N} \approx_{\eps} D-D\sum_{i=0}^{N}\gamma_{i}(\Di A)^i$ of a random walk Laplacian matrix polynomial such that $\gamma$ is a probability distribution over $[0:N]$, it holds for every subset $S\subset V$ that
\[
\xi_{S}^{\rot}(D-\hA_N)\xi_{S}\in[\,(1\pm\eps)\cdot\mathbb{E}_{v\sim\pi_S}[\mathrm{gEsc}(v,S,\mGg)]\,],\quad \text{where}\quad\xi_S\triangleq\mat{1}_{S}/\sqrt{\mu(S)}.
\]
\end{lem}

Using Corollary~\ref{cor:denseTmtx} and Lemma~\ref{lem_mulApproxEGEP}, we propose the first efficient and parallel algorithm that runs in nearly linear work and poly-logarithmic depth that yields a multiplicative approximation of E.G.E.P. for Markov chains with transition matrices induced by $\MDBD$.

\subsection{Faster $\SDDM$ Solver}\label{subsec:FSSMS}

Spielman and Peng~\cite{PS14} gave the first parallel $\SDD$ solver that constructs in $O(m\log^{c_1}n\cdot\log^{3}\kappa)$ work and $O(\log^{c_2}n\cdot\log\kappa)$ depth a sparse $O(1)$-approximate inverse chain that solves approximately to any $\eps>0$ precision an $\SDD$ system in $O((m+n \log^{c}n\cdot\log^{3}\kappa)\log1/\eps)$ work and $O(\log n \cdot\log \kappa\cdot\log 1/\eps)$ depth. In Section~\ref{sec:SDDMSolver}, we give a simple parallel preprocessing step that strengthens their algorithm.

\begin{thm}\label{thm_SDDM_Solver}
There is an algorithm that on input an $n$-dimensional $\SDDM$ matrix $M$ with $m$ non-zeros and condition number at most $\kappa$, produces with probability at least $1/2$ a sparse $O(1)$-approximate inverse chain that can be used to solve any linear equation in $M$ to any precision $\eps>0$ in $O(n \log^{c}n\cdot\log^{3}\kappa\cdot\log1/\eps)$ work and $O(\log n \cdot\log \kappa\cdot\log 1/\eps)$ depth, for some constant $c$. The algorithm runs in $O(m\log^{c_1}n)$ work and $O(\log^{c_2}n\cdot\log\kappa)$ depth for some other constants $c_1,c_2$.
\end{thm}

For the current state-of-the-art result on parallel $\SDD$ solvers we refer the reader to the work of Lee et al.~\cite{KLPSS15}.

\section{Algorithmic Background on Spectral Sparsification}\label{sec:BN}

We write $X\approx_{\eps_{1}\oplus\eps_{2}}Y$ to indicate $(1-\eps_{1})(1-\eps_{2})Y\preceq X\preceq(1+\eps_{1})(1+\eps_{2})Y$. Our analysis uses the following five basic facts (c.f.~\cite{ST14,BGHNT06}).

\begin{fact}\label{fact_fiveProps}
For positive semi-definite (PSD) matrices
$X,Y,W$ and $Z$ it holds

a. if $Y\approx_{\eps}Z$ then $X+Y\approx_{\eps}X+Z$;

b. if $X\approx_{\eps}Y$ and $W\approx_{\eps}Z$ then $X+W\approx_{\eps}Y+Z$;

c. if $X\approx_{\eps_{1}}Y$ and $Y\approx_{\eps_{2}}Z$
then $X\approx_{\eps_{1}\oplus\eps_{2}}Z$;

d. if $X$ and $Y$ are invertible matrices such that $X\approx_{\eps}Y$
then $X^{-1}\approx_{2\eps}Y^{-1}$, $\forall \eps\in(0,\frac{1}{2})$;

e. for any matrix $V$ if $X\approx_{\eps}Y$ then $V^{\rot}XV\approx_{\eps}V^{\rot}YV$.
\end{fact}

\subsection{Prior Algorithms}

Our algorithms for computing spectral sparsifiers of matrix-polynomials use as a black-box several spectral sparsification algorithms for Laplacian and $\SDDM$ matrices.

More precisely, our sequential algorithms build upon Theorem~\ref{thm_KL13} that relies on Kelner and Levin's~\cite[Theorem 3]{KL13} and Cohen et al.'s~\cite[Lemma 4]{CLMMPS15}, and Theorem~\ref{thm_PS14} proposed by Spielman and Peng~\cite[Corollary 6.4]{PS14}.

\begin{thm}\label{thm_KL13}\cite{KL13}
There is an algorithm $\iSS$ takes
as input parameter $\eps\in(0,1)$, matrices $D$ and $A$ such that $D$ is positive diagonal and $A$ is symmetric non-negative with $A_{ii}=0$ for all $i$ such that $L=D-A$ is Laplacian matrix. Then in $\wO(m_{L}\log^{2}n)$ time outputs a positive diagonal matrix $\wD$ and symmetric non-negative matrix $\wA$ such that $nnz(\wA)\leq O(\eps^{-2}n\log n)$,
$\wA_{ii}=0$ for all $i$, and $\wD-\wA\approx_{\eps}D-A$. Moreover, $\wD-\wA$ is Laplacian matrix.
\end{thm}

\begin{thm}\label{thm_PS14}\cite{PS14}
There is an algorithm $\PS$ that takes
as input $\SDDM$ matrix $B=D-M$ and parameter $\eps\in(0,1)$.
Then in $O(\eps^{-2}m_{B}\log^{2}n)$ time outputs
a positive diagonal matrix $\wD$ and symmetric non-negative
matrix $\wM$ with $nnz(\wM)\leq O(\eps^{-2}m_{B}\log n)$
and $\wM_{ii}=0$ for all $i$, such that $\wD-\wM\approx_{\eps}D-M\Di M$ and $\wD\approx_{\eps}D$. Moreover, $\wD-\wM$ is $\SDDM$ matrix.
\end{thm}

Our parallel algorithm uses Theorem~\ref{thm_ParallelLaplacianSS}, which is the culmination of a research line conducted by Spielman and Teng~\cite{ST11}, Orecchia and Vishnoi~\cite{OV11} and Spielman and Peng~\cite{PS14}.

\begin{thm}\label{thm_ParallelLaplacianSS}\cite{PS14}
There is an algorithm $\STOVP$ takes as input a Laplacian matrix $D-M$ and parameter $\eps\in(0,1/2)$. Then it outputs a spectral sparsifier $D-\wM\approx_{\eps}D-M$ with $\wM_{ii}=0$ for all $i$ and $nnz(\wM)\leq O(\eps^{-2}n\log^{c}n)$ for some constant $c$. Moreover, this algorithm requires $O(m\log^{c_1}n)$ work and $O(\log^{c_2}n)$ depth, for some other constants $c_1$ and $c_2$.
\end{thm}

Based on Theorem~\ref{thm_ParallelLaplacianSS} Spielman and Peng~\cite{PS14} parallelized algorithm $\PS$ (c.f. Theorem~\ref{thm_PS14}).

\begin{thm}\label{thm_pPS14}\cite{PS14}
There is a parallel algorithm that on input an $\SDDM$ matrix $D-M$ and parameter $\eps\in(0,1/2)$, outputs a spectral sparsifier $D-\wM\approx_{\eps}D-M\Di M$ with $\wD\approx_{\eps}D$, $\wM_{ii}=0$ for all $i$ and $nnz(\wM)\leq O(\eps^{-2}n\log^{c}n)$ for some constant $c$. Moreover, this algorithm requires $O(\eps^{-2}m\log^{c_1+1}n)$ work and $O(\log^{c_2}n)$ depth, for some other constants $c_1$ and $c_2$.
\end{thm}

\subsection{Spectral Sparsification of $\GL$-Matrices}\label{subsec:SSGL}

We show that the algorithms $\iSS$ and $\PS$ can be amended to produce $\GL$-matrix sparsifiers that are in \emph{normalized form}, i.e. the sparsifiers are expressed in terms of the diagonal matrix $D$ minus a symmetric non-negative matrix $\hM$. Our analysis relies on several results established by Peng et al.~\cite{PengPhd13,PS14,CCLPT14,CCLPT15}.

\begin{lem}\label{lemSSDA}
There is an algorithm $\mSS$ that takes as input a positive diagonal matrix $D$, symmetric non-negative matrix $A$ $(\text{possibly }A_{ii}\neq0)$ such that $B=D-A$ is Laplacian matrix and parameter $\eps\in(0,1)$. Then it outputs in $\wO(m_B\log^{2}n)$
time a spectral sparsifier $D-\hA\approx_{\eps}D-A$ that is Laplacian matrix and satisfies $\hA$ is symmetric non-negative matrix with $nnz(\hA)\leq O(\eps^{-2}n\log n)$.
\end{lem}

The next result implicitly appears in~\cite{CCLPT15}. For completeness we prove it in Appendix \ref{appsec:GLM}.

\begin{lem}\label{lem_my_DappToD}
Suppose $D-A$ is Laplacian matrix $(\text{possibly }A_{ii}\neq0)$ and $\wD-\wA$ a sparsifier with $\wA_{ii}=0$ for every $i$ such that $(1-\eps)(D-A)\preceq\wD-\wA\preceq(1+\eps)(D-A)$.
Then the symmetric non-negative matrix $\hA=(D-\frac{1}{1+\eps}\wD)+\frac{1}{1+\eps}\wA$
satisfies $(1-2\eps)(D-A)\preceq D-\hA\preceq(1+2\eps)(D-A)$.
\end{lem}

We present now the proof of Lemma \ref{lemSSDA}.

\begin{proof}[Proof of Lemma~\ref{lemSSDA}] Notice that $D-A=D^\prm-A^\prm$, where $D^\prm$ is positive diagonal matrix and $A^\prm$ is symmetric non-negative matrix such that $A_{ii}^\prm=0$ for all $i$.
By Theorem \ref{thm_KL13} we obtain a sparsifier $\wD^\prm-\wA^\prm\approx_{\eps/2}D^\prm-A^\prm$. Then by Lemma \ref{lem_my_DappToD} we have  $D^\prm-\hA^\prm\approx_{\eps}D^\prm-A^\prm$, where $\hA^\prm=(D^\prm-\frac{1}{1+\eps}\wD^\prm)+\frac{1}{1+\eps}\wA^\prm$ is symmetric non-negative matrix. We define by $D_{A}=D-D^\prm$ a non-negative diagonal matrix. Set $\hA=D_{A}+\hA^\prm$ and observe that it is symmetric and non-negative matrix. Now the statement follows since $D-\hA=D^\prm-\hA^\prm$.
\end{proof}

\paragraph*{$\GL$-Matrices}
Building upon the work of Spielman and Peng~\cite[Proposition 5.6]{PS14} and Cheng et al.~\cite[Proposition 25]{CCLPT15}, we prove in Appendix~\ref{appsec:GLM} the following statement.

\begin{lem}[Closure]\label{lem_Closure}
Suppose $D-M$ is $\GL$-matrix. Then $D-D(\Di M)^{N}$ is $\GL$-matrix for every $N\in\N_{+}$. Moreover, if $D-\hM\approx_{\eps}D-M$ is a spectral sparsifier, then $D-D(\Di \hM)^{N}$ is $\GL$-matrix for every $N\in\N_{+}$.
\end{lem}

\paragraph*{Normalized Algorithms} We present now two algorithms that sparsify matrices of the form $D-D(\Di M)^N$ for $N\in\{1,2\}$ such that the resulting sparsifiers are in normalized form.

\begin{lem}[Normalized Spectral Sparsification]\label{lem_SS_GL}
There is an algorithm $\mKLC$
that takes as input $\GL$-matrix $B=D-M$ and parameter
$\eps\in(0,1)$, then it outputs in $\wO(m_{B}\log^{2}n)$
time a spectral sparsifier $D-\hM\approx_{\eps}D-M$ that is $\GL$-matrix and $\hM$
is symmetric non-negative matrix with $nnz(\hM)\leq O(\eps^{-2}n\log n)$.
\end{lem}

\begin{proof}
By definition $B=D_{1}+L$ where $D_{1}$ is non-negative diagonal matrix and $L=D_{2}-M$ is Laplacian matrix. We obtain by Lemma \ref{lemSSDA} a sparsifier $D_{2}-\hM\approx_{\eps}D_{2}-M$ that is Laplacian matrix. Now we consider two cases. If $D_1=0$ then we are done. Otherwise $D_1$ is PSD matrix and by Fact \ref{fact_fiveProps}.a we have $D-\hM\approx_{\eps}D-M$. Since $D-M$ is $\SDDM$ matrix and the operator $\approx_{\eps}$ preserves the kernel space, it follows that $D-\hM$ is $\SDDM$ matrix.
\end{proof}

We proceed by stating an interesting structural result that implicitly appears in~\cite{PS14} (c.f. Section ``Efficient
Parallel Construction''). For completeness we prove it in Appendix \ref{appsubsec:SR}.

\newcommand{\lemGLStruct}
{
Suppose $B=D-M$ is $\GL$-matrix. Let
$\eta_{i}=M_{i,:}^{\rot}-M_{i,i}\cdot\mathbf{1}_{i}$ be a column vector,
$d_{i}=\langle M_{i,:},\mathbf{1}\rangle $ and $s_{i}=d_{i}-M_{ii}$
numbers, and $\mD_{N_i}=(s_{i}/d_{i})\cdot\diag(N_i)$
positive diagonal matrix for all $i$, where $N_i=\{ M_{ij}\,|\, M_{ij}\neq0\}$. Let $\mB_{ij}=(M_{ii}/d_{i}+M_{jj}/d_{j})\cdot M_{ij}$ be the $(i,j)$th entry of a matrix with same dimensions as matrix $M$ and $\mD_{B}=\mathrm{diag(\mB\cdot\mathbf{1})}$ be a diagonal matrix.

Then it holds that
$D-M\Di M = \mD_{1}+\mL_{B}+{\sum}_{i=1}^{n}\mL_{N_i}$ where $\mD_{1}=\diag([D-M\Di M]\mathbf{1})$ is non-negative diagonal matrix, $\mL_{B}=\mD_{B}-\mB$ is Laplacian matrix with
at most $m_B$ non-zero entries and every $\mL_{N_i}=(s_{i}/d_{i})\mD_{N_i}-\eta_{i}\eta_{i}^{\rot}/d_{i}$ is Laplacian matrix corresponding to a clique with positively weighted edges that is induced by the neighbour set $N_i$.
}
\mylemma{lemGLStruct}{\lemGLStruct}

Spielman and Peng~\cite{PS14} gave algorithm $\PS$ (c.f. Theorem \ref{thm_PS14}) for sparsifying matrices of the form $D-M\Di M$, where $D-M$ is $\SDDM$ matrix. We extend their result to $\GL$-matrices and our algorithm outputs a spectral sparsifier in normalized form.

\begin{lem}[Normalized 2-Hops Spectral Sparsification]\label{lemmPS}
There is an algorithm $\mPS$ that on input a $\GL$-matrix $B=D-M$ and parameter $\eps\in(0,1)$, outputs in $\wO(\eps^{-2}m_{B}\log^{3}n)$ time a spectral sparsifier $D-\hM \approx_{\eps}D-M\Di M$ that is $\GL$-matrix and $\hM$ is symmetric non-negative matrix with $nnz(\hM )\leq O(\eps^{-2}n\log n)$.
\end{lem}

\begin{proof}
By Lemma \ref{lemSSDA} we have $D-M\Di M=\mD_{1}+\mL$, where $\mD_{1}$ is non-negative diagonal matrix and $\mL$ is sum of Laplacian matrices. Using similar arguments as in ``Section 6 Efficient Parallel Construction''~\cite{PS14} we find a sparsifier $\wD-\wM\approx_{\eps/2}\mL$. Moreover, we can compute the positive diagonal matrix $D^\prm=\diag(\mathbf{L})$ in $O(m_B)$ time (c.f. Appendix~\ref{appsubsec:SR}), and then by Lemma \ref{lem_my_DappToD} we obtain a sparsifier $D^{\prm}-\hM\approx_{\eps}\mL$. Since $\mD_{1}$ is PSD matrix the statement follows by Fact \ref{fact_fiveProps}.a.
\end{proof}

\section{Core Iterative Algorithm}\label{sec:CIA}

Our goal now is to prove Theorem~\ref{thmMyGLN}. We argue in a similar manner as in~\cite{arxivCCLPT15}, but in contrast our analysis shows that the initial sparsification step tolerates higher approximation error. This observation yields an improved algorithm whose runtime is faster by a $O(\log^{2}N)$-factor.

Moreover, we prove that for any $\GL$-matrix $D-M$ such that $M$ is $\SPSD$ matrix, one can construct a spectral sparsifier $D-\hM_2\approx D-M\Di M$ by first computing $D-\hM\approx D-M$ and then $D-\hM_2\approx D-\hM\Di\hM$. This demonstrates that when $M$ is $\SPSD$ matrix, the runtime is dominated by the initial sparsification.

The rest of this section is organized as follows. In Subsection~\ref{subsec:Init_PrwSS} we describe the initial phase of algorithm $\PwrSS$. Then in Subsection~\ref{subsec:IterConst}, we present the iterative construction of the desired spectral sparsifier $D-\hM_{N}\approx_{\eps}D-D(\Di M)^{N}$.

\subsection{Initialization}\label{subsec:Init_PrwSS}

We begin by extending~\cite[Lemma 4.3 and 4.4]{arxivCCLPT15}. For completeness, we provide a prove in Appendix~\ref{appsec:ASC} where in addition we generalize~\cite[Fact 4.2]{arxivCCLPT15}.

\newcommand{\lemSchurRec}
{
Suppose $B=D-M$ is $\GL$-matrix
and $D-\hM \approx_{\eps}D-M$ is a spectral sparsifier.
If $M$ is $\SPSD$ matrix then it holds that $D-\hM \Di \hM \approx_{\eps}D-M\Di M$.
}
\mylemma{lemSchurRec}{\lemSchurRec}

Based on \lemref{lemSchurRec}, we give a faster sparsification algorithm for $\GL$-matrices $D-M\Di M$ such that $M$ is $\SPSD$ matrix.

\begin{lem}\label{lem_my_GL}
There is an algorithm $\fSS$
that takes as input $\GL$-matrix $B=D-M$ such that
$M$ is $\SPSD$ matrix, and parameter $\eps\in(0,1)$.
Then it outputs in $\wO(m_{B}\log^{2}n+\eps^{-4}n\log^{4}n)$
time a spectral sparsifier $D-\hM_{2}\approx_{\eps}D-M\Di M$
that is $\GL$-matrix with $nnz(\hM_{2})\leq O(\eps^{-2}n\log n)$.
\end{lem}

\begin{proof}[Proof of Lemma~\ref{lem_my_GL}] We apply Lemma \ref{lem_SS_GL} to obtain
a sparsifier $D-\hM \approx_{\eps/4}D-M$ in $\wO(m_{B}\log^{2}n)$
time with $nnz(\hM )\leq O(\eps^{-2}n\log n)$
such that $D-\hM $ is $\GL$-matrix. Then by \lemref{lemSchurRec} we know that $D-\hM \Di \hM \approx_{\eps/4}D-M\Di M$.
Now, by Lemma~\ref{lem_Closure} $D-\hM \Di \hM $
is $\GL$-matrix. Then we apply Lemma~\ref{lemmPS}
to obtain in $\wO(\eps^{-4}n\log^{4}n)$
time a sparsifier $D-\hM_{2}\approx_{\eps/4}D-\hM \Di \hM $
with $nnz(\hM )\leq O(\eps^{-2}n\log n)$
such that $D-\hM_{2}$ is $\GL$-matrix. The claims follows by Fact~\ref{fact_fiveProps}.c.
\end{proof}

\subsection{Iterative Construction}\label{subsec:IterConst}

Our analysis of the incurred approximation error after $O(\log N)$ consecutive square sparsification operations builds upon~\cite[Lemma 4.1]{arxivCCLPT15}. In contrast, we prove that for the initial and the final sparsifiers it suffices to have only an $\eps$ approximation, while all intermediate spectral sparsifiers require finer $\eps^{\prm}=\Omega(\eps/\log N)$ approximation. Due to this higher initial error tolerance, we improve the runtime of their algorithm by a $O(\log^{2}N)$-factor.

\begin{lem}[Accumulative Error]\label{lem_IndSS}
Let $D-M$ and $D-\hM_2$ be $\GL$-matrices such that $D-\hM_2\approx_{\eps}D-M\Di M$ and $nnz(\hM_2)\leq O(\eps^{-2}n\log n)$. There is an algorithm $\IndSS$ that on input $\GL$-matrix $D-\hM_2$, integer $N=2^k$ for $k\in\N_+$ and parameter $0<\eps^{\prm}\leq\eps$, outputs in time $\wO({\eps^{\prm}}^{-4}n\log^{4}n\cdot\log N)$ a symmetric non-negative matrix $\hM_{N}$ with
$nnz(\hM_N)\leq O(\eps^{-2}n\log n)$
such that $D-\hM_N\approx_{(\oplus^{(\log N - 1)}\eps^{\prm})\oplus^{2}\eps}D-D(\Di M)^N$ is $\GL$-matrix.
\end{lem}

Our goal now is to prove Lemma~\ref{lem_IndSS}. We establish next a useful algebraic property that all matrices of the form $D(\Di M)^{2^{k}}$ have in common.

\begin{lem}\label{lem_DM_2k}
If $M$ is symmetric matrix, then $D(\Di M)^{2^{k}}$
is $\SPSD$ matrix for every $k\in\N_{+}$.
\end{lem}

\begin{proof}
Let $Y\triangleq\Dhm M \Dhm$. Notice that $D(\Di M)^{2^{k}}=D^{1/2}Y^{2^{k}}D^{1/2}=X^{\rot}X$,
where $X=Y^{2^{k-1}}D^{1/2}$. The statement follows since $X^{\rot}X$
is $\SPSD$ matrix.
\end{proof}

We present now the main iterative procedure used in algorithm $\IndSS$.

\begin{lem}[Iterative Procedure]\label{lem_SqrSSREC}
Let $D-M$ and $D-\hM_{2^{k}}$ be $\GL$-matrices such that $D-\hM_{2^{k}}\approx_{\eps}D-D(\Di M)^{2^{k}}$, for $k\in\N_+$. There is an algorithm $\SqrSS$
that takes as input the $\GL$-matrix $D-\hM_{2^{k}}$
and parameter $\eps^{\prm}\in(0,1)$, then it outputs in $\wO({\eps^{\prm}}^{-2}nnz(\hM_{2^{k}})\log^{3}n)$
time a symmetric non-negative matrix $\hM_{2^{k+1}}$ with
$nnz(\hM_{2^{k+1}})\leq O({\eps^{\prm}}^{-2}n\log n)$
such that $D-\hM_{2^{k+1}}\approx_{\eps\oplus\eps^{\prm}}D-D(\Di M)^{2^{k+1}}$ is $\GL$-matrix.
\end{lem}

\begin{proof}
By Lemma \ref{lem_DM_2k}, $D(\Di M)^{2^{k}}$ is $\SPSD$
matrix for any $k\in\N_{+}$. By Lemma~\ref{lem_Closure} both $D-D(\Di M)^{2^{k}}$ and $D-\hM_{2^{k}}\Di \hM_{2^{k}}$
are $\GL$-matrices. Hence, by \lemref{lemSchurRec}
we have that $D-\hM_{2^{k}}\Di \hM_{2^{k}}\approx_{\eps}D-D(\Di M)^{2^{k+1}}$. Now by Lemma~\ref{lemmPS} we have $D-\hM_{2^{k+1}}\approx_{\eps^{\prm}}D-\hM_{2^{k}}\Di \hM_{2^{k}}$ and hence the statement follows by Fact~\ref{fact_fiveProps}.c.
\end{proof}

Based on the preceding results we are ready to prove Lemma~\ref{lem_IndSS}.

\begin{proof}[Proof of Lemma~\ref{lem_IndSS}]
By Theorem~\ref{thm_PS14} in time $\wO((\eps^{\prm}\cdot\eps)^{-2}n\log^{4}n)$ we can compute a spectral sparsifier $D-\hM_4\approx_{\eps\oplus\eps^{\prm}}D-D(\Di M)^4$ with $nnz(\hM_4)\leq O({\eps^{\prm}}^{-2}n\log n)$. Then we apply $(\log N - 1)$ times Lemma~\ref{lem_SqrSSREC} to obtain in $\wO({\eps^{\prm}}^{-4}n\log^{4}n\cdot\log N)$ time a spectral sparsifier $D-\hM_N\approx_{(\oplus^{(\log N - 1)}\eps^{\prm})\oplus\eps}D-D(\Di M)^N$ with $nnz(\hM_N)\leq O({\eps^{\prm}}^{-2}n\log n)$. The statement follows by applying Theorem~\ref{thm_KL13} with $\eps$ to compute a refined spectral sparsifier of $D-\hM_N$.
\end{proof}

\paragraph*{Proof of Theorem~\ref{thmMyGLN}}
In the initial phase we compute a sparsifier $D-\hM_{2}\approx_{\eps/4}D-M\Di M$ with $nnz(\hM_{2})\leq O(\eps^{-2}n\log n)$ using either Lemma~\ref{lem_my_GL} (when $M$ is $\SPSD$ matrix) or Lemma~\ref{lemmPS}. The statement follows by applying Lemma~\ref{lem_IndSS} with $\eps^{\prm}=\Omega(\eps/\log N)$ to the sparsifier $D-\hM_{2}$.

\section{Spectral Sparsification of Binomial $\GL$-Matrix Polynomials}\label{sec:SSBGLPM}

In this section, we prove Theorem~\ref{thm_LazySS}. We analyze first the properties of matrices of the form $W_{p}=(1-p)I+p\Di M$ for $p\in(0,1)$. It is convenient to associate with them matrix-polynomials $f_{p}(x)=(1-p)+px$ such that $f_{p}(\Di M)=W_{p}$. Since the coefficients of the matrix-polynomial $[f_{p}(x)]^{N}$ follow Binomial distribution $B(N,p)$, it follows that
$W_{p}^{N} = \sum_{i=0}^{N}B_{N,i}(p)\cdot(\Di M)^{i}$ for every $p\in(0,1)$, where $B_{N,i}(p)={N \choose i}p^{i}(1-p)^{N-i}$.

When $D-M$ is a Laplacian matrix, the matrix $W_{p}^N$ corresponds to the transition matrix of a $p$-lazy random walk process of length $N$ (c.f.~\cite{SJ89}). We associate to such a Markov chain a matrix-polynomial $D-DW_{p}^{N}=D - D\sum_{i=0}^{N} B_{N,i}(p) \cdot(\Di M)^{i}$. We present now some useful algebraic properties of matrices of the form $DW_{p}^{2^{k}}$ and $D-DW_{p}^{N}$.

\newcommand{\lemDDWWp}
{
Suppose $D-M$ is $\GL$-matrix. Then
$D-DW_{p}^{N}$ is $\GL$-matrix for every $N\in\N_{+}$. Also $DW_{p}$
is $\SPSD$ matrix $\forall p\in(0,1/2]$ and $DW_{p}^{2^{k}}$
is $\SPSD$ matrix $\forall p\in(0,1)$ and $\forall k\in\N_{+}$.
}
\mylemma{lemDDWWp}{\lemDDWWp}

\begin{proof}
By definition of $W_{p}$, we have $D-DW_{p}=p(D-M)$ is $\GL$-matrix. Suppose $D-M$ is Laplacian matrix, then $D-DW_{p}$ is Laplacian matrix and by Lemma~\ref{lem_Closure}, $D-DW_{p}^{N}$ is Laplacian matrix for every $N\in\N_{+}$. Suppose now that $D-M$ is $\SDDM$ matrix, then $D-DW_{p}$ is $\SDDM$ matrix and by Lemma \ref{lem_SDDM_Closure} $D-DW_{p}^{N}$ is $\SDDM$ matrix for every $N\in\N_{+}$.

By definition $DW_{p}=(1-p)D+pM$
and since $D-M$ is diagonally dominant, it holds that
$DW_{p}$ is $\SPSD$ matrix for every $p\in(0,1/2]$.
Moreover, since $DW_{p}$ is symmetric matrix by Lemma \ref{lem_DM_2k} it holds that $D(\Di\cdot DW_{p})^{2^{k}}=DW_{p}^{2^{k}}$ is $\SPSD$ matrix for every $k\in\N_{+}$.
\end{proof}

\begin{proof}[Proof of Theorem~\ref{thm_LazySS}] The statement follows by \lemref{lemDDWWp} and Theorem~\ref{thmMyGLN}.
\end{proof}

\section{Spectral Sparsification of $\GL$-Matrix Polynomials Induced by $\MDBD$}\label{sec:SSGLBM}

Here we prove Theorem~\ref{thm_SS_MGL}. Our approach relies on the following key algorithmic idea.

\begin{lem}[Preprocessing of 2-Hop Spectral Sparsification]\label{lem_DMp2}
Let $D-M$ be a $\GL$-matrix, $D-\hM_{1}\approx_{\eps}D-M$ and $D-\hM_{2}\approx_{\eps}D-M\Di M$ are spectral sparsifiers. Then for every $p\in(0,1)$ the sparse matrix $\hM_{p,2}=(1-p)^{2}D+2(1-p)p\hM_{1}+p^{2}\hM_{2}$ yields a spectral sparsifier $D-\hM_{p,2}\approx_{\eps}D-DW_{p}^{2}$ that is $\GL$-matrix.
\end{lem}

\begin{proof}
The statement follows by
\begin{eqnarray*}
  D-\hM_{p,2} & = & 2p(1-p)[D-\hM_{1}]+p^{2}[D-\hM_{2}]\\
  & \approx_{\eps} & 2p(1-p)[D-M]+p^{2}[D-M\Di M]=D-DW_{p}^{2}.
\end{eqnarray*}
\end{proof}

Our algorithm $\SSMDBD$ builds upon Lemma~\ref{lem_DMp2} and Lemma~\ref{lem_IndSS}. Due to the preprocessing step in Lemma~\ref{lem_DMp2} we speed up the sparsification of each $\GL$-matrix polynomial $D-DW_{p_j}^{N}$ for all $j\in[1:T]$. We present now the pseudo code of algorithm $\SSMDBD$.

\begin{algorithm}[H]
\caption{}
\label{alg_SSMDBD}

$(D,\hM )=\SSMDBD(D,M,\BNTap,\eps)$

1. Let $\eps^{\prm}=\eps/[3\log N]$, $\hM_{tmp}=0$ and  $\delta=1-\sum_{i=1}^{T}\alpha_i$.

2. $(D,\hM_{1},\hM_{2})=\InitSS(D,M,\eps/3)$.\\
3. For every $j\in\{ 1,\dots,T\} $ do

$\qquad$3.1 Set $p=p_{j,T}=j/(T+1)$ and $\hM_{p,2}=(1-p)^{2}D + 2p(1-p)\hM_{1} + p^{2}\hM_{2}$.

$\qquad$3.2 $(D,\hM_{p,N})=\IndSS(\hM_{p,2},N,\eps^{\prm})$
where $D-\hM_{p,N}\approx_{2\eps/3}D-DW_{p}^{N}$ (c.f. Lemma~\ref{lem_IndSS}).

$\qquad$3.3 $\hM_{tmp}=\hM_{tmp}+\alpha_{j}\cdot\hM_{p,N}$.

4. Sparsify $D-\hM \approx_{\eps/3}D - \frac{1}{1-\delta}\hM_{tmp}$
by algorithm $\mKLC$ (c.f. Lemma \ref{lem_SS_GL}).

5. Return $(D,\hM)$.
\end{algorithm}

\begin{algorithm}[H]
\caption{}
\label{alg_Preprocess}

$(D,\hM_{1},\hM_{2})=\InitSS(D,M,\eps)$

1. Sparsify $D-\hM_{1} \approx_{\eps}D-M$ by algorithm $\mKLC$ (c.f. Lemma \ref{lem_SS_GL}).

2. Sparsify $D-\hM_{2} \approx_{\eps}D-M\Di M$

$\qquad$2.1 If $M$ is $\SPSD$ matrix call algorithm $\fSS$ (c.f. Lemma~\ref{lem_my_GL}),

$\qquad$2.2 otherwise call algorithm $\mPS$ (c.f. Lemma~\ref{lemmPS}).

3. Return $(D,\hM_{1},\hM_{2})$.
\end{algorithm}

Let $\delta=1-\sum_{i=1}^{T}\alpha_i$. We denote a $\GL$-matrix polynomial induced by $\MDBD$ $\BNTap$ as
$\GLB \triangleq\sum_{j=1}^{T}\alpha_{j}(D-DW_{p_j}^{N}) =
(1-\delta)D-D\sum_{i=0}^{N}\gamma_{i}(\Di M)^{i}$,
where $\gamma_{i}=\sum_{j=1}^{T}\alpha_{j}B_{N,i}(p_j)$. We are now ready to prove Theorem~\ref{thm_SS_MGL}.

\begin{proof}[Proof of Theorem \ref{thm_SS_MGL}] Let $\eps^{\prm}=\eps/[4\log N]$.
We perform first a preprocessing step. We apply Lemma \ref{lem_SS_GL}
to obtain a sparsifier $D-\hM \approx_{\eps^{\prm}}D-M$.
Then depending on whether $M$ is $\SPSD$ matrix
we use either Lemma~\ref{lemmPS} or Lemma~\ref{lem_my_GL}
to obtain a sparsifier $D-\hM_{2}\approx_{\eps^{\prm}}D-M\Di M$.
The run time is at most $\wO(\eps^{-2}m_{B}\log^{3}n\cdot\log^{2}N)$ or $\wO(m_{B}\log^{2}n+\eps^{-4}n\log^{4}n)$ respectively. Moreover, the sparsifiers satisfy $nnz(\hM_{1}),nnz(\hM_{2})\leq O(\eps^{-2}n\log n\cdot\log^{2}N)$.

We combine Lemma \ref{lem_DMp2} and Theorem~\ref{thmMyGLN}
to find each sparsifier $D-\hM_{p_{j},N}\approx_{\eps}D-DW_{p_{j}}^{N}$
by initializing algorithm $\PwrSS$ with a sparsifier $D-\hM_{p_{j},2}\approx_{\eps^{\prm}}D-DW_{p_{j}}^{2}$.
Let $\hM_{tmp} = \sum_{j=1}^{T}\alpha_{j}\hM_{p_{j},N}$, then by Fact~\ref{fact_fiveProps}.b it holds $(1-\delta)D-\hM_{tmp} \approx_{\eps}\GLB$.
This phase has $\wO(\eps^{-4}nT\log^{4}n\cdot\log^{5}N)$ runtime and each sparsifier satisfies $nnz(\hM_{p_{i},N})\leq O(\eps^{-2}n\log n)$.

However, matrix $\hM_{tmp}$ can be dense. We
apply Lemma \ref{lem_SS_GL} to obtain a sparsifier $D-\hM \approx_{\eps}D-\frac{1}{1-\delta}\hM_{tmp} \approx_{\eps} \frac{1}{1-\delta}\GLB$
in time $\wO(\min\{ n^{2}\log^{2}n,\eps^{-2}nT\log^{3}n\cdot\log^{2}N\} )$
such that $nnz(\hM )\leq O(\eps^{-2}n\log n)$.
\end{proof}

\section{Parallelization of Algorithm $\SSMDBD$}\label{sec:effPrlAlg}

In this section we parallelize algorithm $\SSMDBD$. This gives the first efficient parallel algorithm that computes a spectral sparsfier of any $\GL$-matrix polynomial with coefficients induced by $\MDBD$.

Our goal now is to prove Theorem~\ref{thm_prl_SS_MGL}. By construction of algorithms $\SSMDBD$ and $\InitSS$, it suffices to show that we can efficiently parallelize algorithms $\mKLC$, $\mPS$ and $\PwrSS$. Then the statement follows by noting that each $\GL$ Binomial matrix-polynomial can be sparsified separately and in parallel.

We parallelize now algorithms $\mKLC$, $\mPS$ and $\IndSS$.

\begin{lem}\label{lem_pKLC}
There is a parallel algorithm $\pKLC$
that on input $\GL$-matrix $B=D-M$ and parameter
$\eps\in(0,1)$, outputs a spectral sparsifier $D-\hM\approx_{\eps}D-M$ that is $\GL$-matrix such that $\hM$
is symmetric non-negative matrix with $nnz(\hM)\leq O(\eps^{-2}n\log^{c} n)$ for some constant $c$. The algorithm runs in $O(m_{B}\log^{c_1}n)$ work and $O(\log^{c_2}n)$ depth, for some other constants $c_1,c_2$.
\end{lem}

\begin{proof}
We argue in a similar manner as in Lemma~\ref{lem_SS_GL} to show that the statement holds for $\GL$-matrices (Laplacian or $\SDDM$ matrices). Then the statement follows by Theorem~\ref{thm_ParallelLaplacianSS}.
\end{proof}

\begin{lem}\label{lem_pPS}
There is a parallel algorithm $\pPS$ that on input $\GL$-matrix $B=D-M$ and parameter $\eps\in(0,1/2)$, outputs a spectral sparsifier $D-\wM\approx_{\eps}D-M\Di M$ that is $\GL$-matrix such that $\wD\approx_{\eps}D$, $\wM_{ii}=0$ for all $i$ and $nnz(\wM)\leq O(\eps^{-2}n\log^{c}n)$ for some constant $c$. The algorithm runs in $O(\eps^{-2}m\log^{c_1+1}n)$ work and $O(\log^{c_2}n)$ depth, for some other constants $c_1$ and $c_2$.
\end{lem}

\begin{proof}
Using similar arguments as in Lemma~\ref{lemmPS} we prove that the statement holds for $\GL$-matrices. Then the statement follows by Theorem~\ref{thm_pPS14}.
\end{proof}

\newcommand{\lempSqrSSREC}
{
Let $D-M$ and $D-\hM_{2^{k}}$
are $\GL$-matrices such that $D-\hM_{2^{k}}\approx_{\eps}D-D(\Di M)^{2^{k}}$, for $k\in\N_+$. There is a parallel algorithm $\pSqrSS$
that on input the $\GL$-matrix $D-\hM_{2^{k}}$
and parameter $\eps^{\prm}\in(0,1)$, outputs a spectral sparsifier $D-\hM_{2^{k+1}}\approx_{\eps\oplus\eps^{\prm}}D-D(\Di M)^{2^{k+1}}$ that is $\GL$-matrix with $nnz(\hM_{2^{k+1}})\leq O({\eps^{\prm}}^{-2}n\log^{c} n)$ for some constant $c$. The algorithm runs in work $O({\eps^{\prm}}^{-2}nnz(\hM_{2^{k}})\log^{c_1+1}n)$ and depth $O(\log^{c_2}n)$, for some other constants $c_1,c_2$.
}
\mylemma{lempSqrSSREC}{\lempSqrSSREC}

\begin{proof}
We use similar arguments as in Lemma~\ref{lem_SqrSSREC}, but we substitute Lemma~\ref{lemmPS} with Lemma~\ref{lem_pPS}.
\end{proof}

\newcommand{\lempIndSS}
{
Let $D-M$ and $D-\hM_2$
are $\GL$-matrices such that $D-\hM_2\approx_{\eps}D-M\Di M$ and $nnz(\hM_2)\leq O(\eps^{-2}n\log^{c} n)$ for some constant $c$. There is a parallel algorithm $\pIndSS$ that on input $\GL$-matrix $D-\hM_2$, integer $N=2^k$ for $k\in\N$ and parameter $0<\eps^{\prm}\leq\eps$, outputs a spectral sparsifier $D-\hM_N\approx_{(\oplus^{(\log N - 1)}\eps^{\prm})\oplus^{2}\eps}D-D(\Di M)^N$ that is $\GL$-matrix and $nnz(\hM_N)\leq O(\eps^{-2}n\log^{c}n)$. The algorithm runs in work $\wO({\eps^{\prm}}^{-4}n\log^{c+c_1+1}n)$ and depth $O(\log^{c_2}n\cdot\log N)$ for some constants $c_1,c_2$.
}
\mylemma{lempIndSS}{\lempIndSS}

\begin{proof}
We argue in a similar manner as in Lemma~\ref{lem_IndSS}. By Lemma~\ref{lem_pKLC} we compute a sparsifier $D-\hM_4\approx_{\eps^{\prm}\oplus\eps}D-D(\Di M)^4$ with $nnz(\hM_4)\leq O({\eps^{\prm}}^{-2}n\log^{c}n)$ in work $O((\eps\cdot\eps^{\prm})^{-2}n\log^{c+c_1+1}n)$ and depth $O(\log^{c_2}n)$. Then we apply $(\log N - 1)$ times \lemref{lempSqrSSREC} to obtain a spectral sparsifier $D-\hM_N\approx_{(\oplus^{(\log N - 1)}\eps^{\prm})\oplus\eps}D-D(\Di M)^N$ with $nnz(\hM_N)\leq O({\eps^{\prm}}^{-2}n\log^{c} n)$ in $O({\eps^{\prm}}^{-4}n\log^{c+c_1+1}n)$ work and $O(\log^{c_2}\cdot\log N)$ depth. The statement follows by applying Lemma~\ref{lem_pKLC} with $\eps$ to compute a refined spectral sparsifier of $D-\hM_N$.
\end{proof}

We present now the proof of Theorem~\ref{thm_prl_SS_MGL} which yields the parallel algorithm $\pSSMDBD$.

\begin{proof}[Proof of Theorem~\ref{thm_prl_SS_MGL}]
We sketch first our parallel algorithm $\pSSMDBD$. We parallelize algorithm $\InitSS$ based on algorithms $\pKLC$ and $\pPS$. Then, we sparsify separately and in parallel each of the $T$ distinct single Binomial $\GL$-matrix polynomials by algorithm $\pIndSS$. The resulting $T$ sparsifiers are scaled and merged into a $\GL$-matrix polynomial induced by $\MDBD$. Since this matrix-polynomial might be dense, we sparsify it using algorithm $\pKLC$.

The correctness of algorithm $\pSSMDBD$ follows by Theorem~\ref{thm_SS_MGL}. We analyze now the work and the depth of algorithm $\pSSMDBD$. By Lemma~\ref{lem_pKLC} and Lemma~\ref{lem_pPS} the initial phase is dominated by $O(\eps^{-2}m\log^{c_1+1}n\cdot\log^{2}N)$ work and $O(\log^{c_2}n)$ depth. Moreover, each of the sparsifiers $D-\hM \approx_{\eps}D-M$ and $D-\hM_{2}\approx_{\eps}D-M\Di M$ has at most $O(\eps^{-2}n\log^{c}n)$ non-zero entries.

Let $\eps^{\prm}=\eps/[12\log N]$. By \lemref{lempIndSS} for each $j\in[1:T]$ we apply algorithm $\pIndSS$ to compute a spectral sparsifier $D-\hM_{p_j,N}\approx_{\eps/6}D-DW_{p_j}^{N}$ with $nnz(\hM_{p_j,N})\leq O(\eps^{-2}n\log^{c}n)$. This phase runs in work $\wO({\eps^{\prm}}^{-4}nT\log^{c+c_1+1}n)$ and depth $O(\log^{c_2}n\cdot\log N)$. Furthermore, the linear combination $\hM_{tmp}=\sum_{j=1}^{T}\alpha_{j}\cdot\hM_{p,N}$ (in Step 3.3) can be computed in depth $O(\log T)$.

Therefore, we approximate by $D-\frac{1}{1-\delta}\hM_{tmp}\approx_{\eps/6} D-D\sum_{i=0}^{N}\frac{1}{1-\delta}\gamma_{i}(\Di M)^{i}$ the desired $\GL$-matrix polynomial induced by $\MDBD$. Since matrix $\hM_{tmp}$ might be dense, by Lemma~\ref{lem_pKLC} we compute a spectral sparsifier $D-\hM \approx_{\eps/12}D - \frac{1}{1-\delta}\hM_{tmp}$ with $nnz(\hM)\leq O(\eps^{-2}n\log^{c_2}n)$. Algorithm $\pKLC$ runs in work
$O(\min\{\eps^{-2}nT\log^{c}n,\,n^2\}\cdot\log^{c_1}n)$ and depth $O(\log^{c_2}n)$.
\end{proof}

\section{Faster $\SDDM$ Solver}\label{sec:SDDMSolver}

In this section we prove Theorem~\ref{thm_SDDM_Solver}. We argue in a similar manner as in~\cite{PS14}, but in contrast our improved analysis relies on the refined initialization phase developed in Section~\ref{sec:CIA} and its consecutive parallelization in Section~\ref{sec:effPrlAlg}. Spielman and Peng's~\cite{PS14} proof involves two major steps: the first is to construct a sparse $O(1)$-approximate inverse chain, and the second is to apply this chain as a preconditioner into an algorithm known as ``Preconditioned Richardson Iteration''~\cite[Lemma 4.4]{PS14}.

We give now an improved construction for a sparse $\eps$-approximate inverse chain. This directly implies the desired statement of Theorem~\ref{thm_SDDM_Solver}.

\begin{proof}[Proof of Theorem~\ref{thm_SDDM_Solver}]

We apply Lemma \ref{lem_pKLC} to compute a sparsifier $D-\hM \approx_{\eps/16}D-M=B$ with $nnz(\hM )\leq O(\eps^{-2}n\log^{c} n)$ in work $O(m_{B}\log^{c_1}n)$ and depth $O(\log^{c_2}n)$. Then by Fact~\ref{fact_fiveProps}.d we have $(D-\hM )^{-1}\approx_{\eps/8}(D-M)^{-1}$.

Our goal now is to construct a sparse $\eps$-approximate inverse chain of the sparsifier $\hB=D-\hM$. The condition number of matrix $\hB$ satisfies $\kappa_{\hB}\leq\frac{1+\eps/8}{1-\eps/8}\kappa_{B}=t_{\hB}$. Let $\eps^{\prm}=\eps/(16\log t_{\hB}$). Spielman and Peng~\cite{PS14} proved that $O(\log t_{\hB})$ iterations suffice for the following iterative procedure to output a sparse $\eps$-approximate inverse chain.

By Lemma~\ref{lem_pPS} we compute a spectral sparsifier $D-\hM_{2}\approx_{\eps/8}D-\hM \Di \hM$ with $nnz(\hM_2)\leq O(\eps^{-2}n\log^{c}n)$ in work $O(\eps^{-4}n\log^{c+c_1+1}n)$ and depth $O(\log^{c_2}n)$. By Lemma~\ref{lem_Closure} and \lemref{lempSqrSSREC} for each consecutive call we compute a sparsifier $D-\hM_{2^{k+1}}\approx_{\eps^{\prm}}D-\hM_{2^k} \Di \hM_{2^k}$ with the at most $O(\eps^{-2}n\log^{c}n\cdot\log^{2}\kappa_B)$ non-zero entries in work $O(\eps^{-4}n\log^{c+c_1+1}n\cdot\log^{2}\kappa_B)$ and depth $O(\log^{c_2}n)$.

We combine now Fact \ref{fact_fiveProps} and apply recursively $O(\log t_{\hB})$ times the relation
\begin{eqnarray*}
& & [\Di+(I+\Di\hM_{2^{k}})[D-\hM_{2^{k+1}}]^{-1}(I+\hM_{2^{k}}\Di)]/2\\ & \approx_{\eps^{\prm}} & [\Di+(I+\Di\hM_{2^{k}})[D-\hM_{2^{k}}\Di\hM_{2^{k}}]^{-1}(I+\hM_{2^{k}}\Di)]/2\\&=&(D-\hM_{2^{k}})^{-1}.
\end{eqnarray*}
Spielman and Peng showed in~\cite[Corollary 5.5]{PS14} that $D-\hM_{2^{k}}\Di\hM_{2^{k}}$ can be replaced with $D$ for $k=O(\log t_{\hB})$ and maintain the desired approximation. The statement follows by Fact~\ref{fact_fiveProps}.
\end{proof}

\section{Representational Power of $\MDBD$}\label{sec:APMDBD}

In this section we prove Theorem~\ref{thmMDBD}. Our analysis relies on the following three influential works. Hald~\cite{Hald68} analyzed mixed Binomial distributions in continuous case. Cruz-Uribe and Neugebauer~\cite{UN03,CUN02} gave sharp guarantees for approximating integrals using the Trapezoid method. Doha et al.~\cite{DBS11} proved a simple closed formula for higher order derivatives of Bernstein basis.

Our goal now is to prove Theorem~Theorem~\ref{thmMDBD}. Hald~\cite{Hald68} proved the following result on mixed Binomial distributions.

\begin{thm}\cite{Hald68}\label{thm_Hald68}
Let $w(x)$ be a probability
density function that is four times differentiable. Then for every $N\in\N$
and $i\in[0,N]$ the Bernstein basis $B_{N,i}(p)$
satisfies
\begin{equation}\label{eq:int_MDBD}
\int_{0}^{1}w(p)\cdot B_{N,i}(p)dp=\frac{w(i/N)}{N}\cdot\left[1+\frac{b_{1}(i/N)}{N}+\frac{b_{2}(i/N)}{N^{2}}+O\left(\frac{1}{N^{3}}\right)\right]
\end{equation}
where the functions are defined by $b_{1}(x)=\frac{1}{w(x)}[-w(x)+(1-2x)w^{\prm}(x)+\frac{1}{2}x(1-x)w^{\prm\prm}(x)]$
and $b_{2}(x)=\frac{1}{w(x)}[w(x)-3(1-2x)w^{\prm}(x)+(1-6x+6x^{2})w^{\prm\prm}(x)+\frac{5}{6}x(1-x)(1-2x)w^{\prm\prm\prm}(x)+\frac{1}{8}x^{2}(1-x)^{2}w^{\prm v}(x)].$
\end{thm}

We distinguish two types of approximation errors. The error term $(1+\eta_{i})$ (c.f. Equation~\ref{eq:fin_Approx} in Theorem~\ref{thmMDBD}) is caused by the error introduced in Equation \ref{eq:int_MDBD}. The second error type is due to the integral discretization with finite summation. The later approximation error in analyzed by Cruz-Uribe and Neugebauer~\cite{UN03,CUN02}. We summarize below their result.

\begin{thm}\cite{CUN02,UN03}\label{thm_approx_Def_Intergral}
Suppose $f$ be continuous
and twice differentiable function, $T\in\N$ is number, and the discrete
approximator of $f$ is defined by $A_{T}(f)=[\frac{1}{2}(f(0)+f(1))+\sum_{i=1}^{T-1}f(i/T)]/T$.
Then the approximation error is given by the expression
$E_{T}(f)=|A{}_{T}(f)-\int_{0}^{1}f(t)dt|=|\sum_{i=1}^{T}L_{i}|$,
where
$L_{i}=\frac{1}{2}\int_{x_{i-1}}^{x_{i}}\left[\frac{1}{4T^{2}}-(t-c_{i})^{2}\right]f^{\prm\prm}(t)dt$ and $c_{i}=(x_{i-1}+x_{i})/2$.
\end{thm}

The rest of this section is devoted to prove Theorem~\ref{thmMDBD}. We use the following two results established by Cruz-Uribe and Neugebauer, and Doha et al.

\begin{lem}
\label{lem_int_n1}~\cite{CUN02,UN03} The Bernstein basis satisfies $\int_{0}^{1}B_{N,i}(x)dx=\frac{1}{N+1}$ for every $i\in[0:N]$.
\end{lem}

\begin{lem}
\label{lem_pder_Bni}~\cite{DBS11} The $p$th derivative of a Bernstein
basis satisfies for every $i\in[0:N]$ that
\[
\frac{d^{p}B_{N,i}(x)}{dx^{p}}=\frac{N!}{(N-p)!}\sum_{k=\max\{ 0,i+p-N\} }^{\min\{ i,p\} }(-1)^{k+p}\cdot{p \choose k}\cdot B_{N-p,i-k}(x).
\]
\end{lem}

We propose an upper bound on the integral of $p$th order derivative of Bernstein basis.
\begin{cor}
\label{cor_integral_derBB} For every $p\in[1:N-1]$ and
$i\in[p+1:N-(p+1)]$ such that $i+p\leq N$,
the Bernstein basis satisfies
\[
\int_{0}^{1}\left|\frac{d^{p}B_{N,i}(x)}{dx^{p}}(t)\right|dt\leq\frac{N!}{(N-p)!}\cdot\frac{2^{p}}{N+1}.
\]
\end{cor}

\begin{proof}
We combine Lemma \ref{lem_int_n1} and Lemma \ref{lem_pder_Bni} to obtain
\[
\int_{0}^{1}\left|\frac{d^{p}B_{N,i}(x)}{dx^{p}}(t)\right|dt\leq\frac{N!}{(N-p)!}\sum_{k=0}^{p}{p \choose k}\cdot\int_{0}^{1}B_{N-p,i-k}(t)dt=\frac{N!}{(N-p)!}\cdot\frac{2^{p}}{N+1}.
\]
\end{proof}

We are now ready to prove Theorem~\ref{thmMDBD}.

\begin{proof}[Proof of Theorem~\ref{thmMDBD}]
Recall that $F_i(x)=w(x)B_{N,i}(x)$. By Theorem \ref{thm_approx_Def_Intergral} we have
\[
\left|\sum_{k=1}^{T}L_{k}\right|=\left|\frac{1}{2}\sum_{k=1}^{T}\int_{x_{k-1}}^{x_{k}}\left[\frac{1}{4T^{2}}-(t-c_{k})^{2}\right]\cdot\frac{d^{2}F_i(x)}{dx^{2}}(t)dt\right|\leq\frac{1}{8T^{2}}\int_{0}^{1}\left|\frac{d^{2}F_i(x)}{dx^{2}}(t)\right|dt.
\]
Since $|\frac{d^{2}F_i(x)}{dx^{2}}|=w^{\prm\prm}\cdot B_{N,i}+2\cdot w^{\prm}\cdot B_{N,i}^{\prm}+w\cdot B_{N,i}^{\prm\prm}$,
we consider following three cases:\\
\textbf{Case 1:} We combine $\max_{x\in[0,1]}|w^{\prm\prm}(x)|\leq 2\phiw\cdot N^{2}$
and Lemma \ref{lem_int_n1} to obtain
\[
\int_{0}^{1}|w^{\prm\prm}(t)\cdot B_{N,i}(t)|dt\leq 2\phiw\cdot N.
\]
\textbf{Case 2:} Using $\max_{x\in[0,1]}|w^{\prm}(x)|\leq\frac{1}{2}\phiw\cdot N$
and Corollary \ref{cor_integral_derBB} it holds
\[
\int_{0}^{1}|w^{\prm}(t)\cdot B_{N,i}^{\prm}(t)|dt\leq \phiw\cdot N.
\]
\textbf{Case 3:} Combining $\max_{x\in[0,1]}|w(x)|\leq \phiw$
and Corollary \ref{cor_integral_derBB} yields
\[
\int_{0}^{1}|w(t)\cdot B_{N,i}^{\prm\prm}(t)|dt \leq \phiw\cdot \int_{0}^{1}B_{N,i}^{\prm\prm}(t)dt \leq 4\phiw\cdot N.
\]
The desired result follows from the preceding three cases and Theorem~\ref{thm_Hald68}.
\end{proof}

\section{Approximating Discretized PDF}\label{sec:approxDiscrPDFs}

In this section we prove Theorem~\ref{thm_AppDscrPDF}. We begin our discussion by presenting the pseudo code of algorithm $\AppDscrPDF$.

\begin{algorithm}[H]
\caption{Approximate Discretized PDF by MDBD}

$\BNTap=\AppDscrPDF(w,\phiw,N,\eps_{I})$

1. Compute $S_{N+1,N}=\sum_{i=0}^{N}w(i/N)$ and $S_{T,T+1}=\sum_{i=1}^{T}w(i/[T+1])$, where $T=\lceil N\sqrt{\phiw/\eps_{I}} \rceil$.

2. Compute $\alpha_{j} = w(p_{j})\cdot N/[(T+1)\cdot S_{N+1,N}]$,
where $p_{j}=j/[T+1]$ for all $j\in[1:T]$.

3. Return $\BNTap$.
\end{algorithm}

Before we prove Theorem~\ref{thm_AppDscrPDF}, we analyze the class of continuous probability density functions that admit a discretized approximation by $\MDBD$.

\begin{lem}
\label{lem_RatioLambda} Let $w(x)=C\cdot f(x)$ be a twice differentiable p.d.f. such that for $N\in\N_+$ it holds $a)\;0\leq f(x)\leq1$,\,\,\,\,\,\,$b)\;\frac{1}{2}[f(0)+f(1)] \geq \Omega(1)$,\,\,\,\,\,\,$c)\;1\leq C\leq o(N)$,\,\,\,\,\,\,$ and $\,\,\,\,\,\,$d)\;\int_{0}^{1}\left|f^{(2)}(x)\right|dx\leq o(N)$.
Then, for $T=\lceil N\sqrt{\phiw/\epsI} \rceil$ with $\phiw/\epsI\geq 1$ it holds
\[
1-\frac{S_{T,T+1}/(T+1)}{S_{N+1,N}/N}=o(1),\,\,\text{where}\,\,\, S_{T,T+1}\triangleq\sum_{k=1}^{T}w(k/[T+1])\,\,\,\text{and}\,\,\, S_{N+1,N}\triangleq\sum_{k=0}^{N}w(k/N).
\]
\end{lem}

\begin{proof}
By Theorem~\ref{thm_approx_Def_Intergral} for the discrete approximator of $f$
\[
A_{M}(f)=\frac{1}{M}\left[\frac{1}{2}\left(f(0)+f(1)\right)+\sum_{i=1}^{M-1}f\left(\frac{i}{M}\right)\right],
\]
it holds that
\[
\left|A_{T+1}(f)-\int_{0}^{1}f(t)dt\right|\leq\frac{1}{8\left(T+1\right)^{2}}\int_{0}^{1}\left|f^{(2)}(t)\right|dt\lesssim\frac{\epsI}{\phiw}\cdot o\left(\frac{1}{N}\right),
\]
and similarly
\[
\left|A_{N}(f)-\int_{0}^{1}f(t)dt\right|\leq\frac{1}{8N^{2}}\int_{0}^{1}\left|f^{(2)}(t)\right|dt\lesssim o\left(\frac{1}{N}\right).
\]
By definition, $\int_{0}^{1}f(t)dt=C^{-1}\in(1/o(N),1]$
and thus
\[
A_{T+1}(f)\in\left[\frac{1}{C}\pm\frac{\epsI}{\phiw}\cdot o\left(\frac{1}{N}\right)\right],\quad\text{and}\quad A_{N}(f)\in\left[\frac{1}{C}\pm o\left(\frac{1}{N}\right)\right].
\]
Let $d\triangleq[f(0)+f(1)]/2$. Straightforward checking shows that
\[
\frac{S_{T,T+1}}{T+1}=C\cdot\left[A_{T+1}(f)-\frac{d}{T+1}\right] \quad\text{and}\quad \frac{S_{N+1,N}}{N}=C\cdot\left[A_{N}(f)+\frac{d}{N}\right].
\]
We prove now the upper bound. By assumption $d\in[\Omega(1),1]$ and since $C\leq o(N)$ we have
\begin{eqnarray*}
 &  & \Lambda \triangleq \frac{S_{T,T+1}/(T+1)}{S_{N+1,N}/N} = \frac{A_{T+1}(f)-\frac{d}{T+1}}{A_{N}(f)+\frac{d}{N}} = \frac{\frac{1}{C}-\frac{d}{T+1}\pm\frac{\epsI}{\phiw}\cdot o\left(\frac{1}{N}\right)}{\frac{1}{C}+\frac{d}{N}\pm o\left(\frac{1}{N}\right)}\\
 & \leq & 1-\frac{\left(1+\frac{1}{2}\sqrt{\frac{\epsI}{k}}\right)\cdot\frac{d}{N} - \left(1+\frac{\epsI}{k}\right)\cdot o\left(\frac{1}{N}\right)}{\frac{1}{C} + \left[\frac{d}{N} - o\left(\frac{1}{N}\right)\right]} \leq 1 - \frac{\Omega(1)}{N}\cdot\frac{1}{\frac{1}{C} + \frac{1}{N}} = 1-o(1).
\end{eqnarray*}
We can prove the lower bound $\Lambda\geq1-o(1)$ using similar arguments.
\end{proof}

We present now the proof of Theorem~\ref{thm_AppDscrPDF}.

\begin{proof}[Proof of Theorem~\ref{thm_AppDscrPDF}]
By definition $S_{T,T+1}=\sum_{k=1}^{T}w(j/[T+1])$, $S_{N+1,N}=\sum_{j=0}^{N}w(j/N)$ and the desired discretized p.d.f. is
\[
\hw(i/N)=\frac{w(i/N)}{\sum_{j=0}^{N}w(j/N)}=\frac{w(i/N)}{S_{N+1,N}}.
\]
By Theorem~\ref{thmMDBD}, it holds for all $i\in\left[3:N-3\right]$ that
\begin{equation}\label{eq:LAD}
\left|\left(1+\eta_{i}\right)\frac{w(i/N)}{N} - \sum_{j=1}^{T}\frac{w(j/[T+1])}{T+1}\cdot B_{N,i}(j/[T+1])\right| \leq \frac{\epsI}{N}.
\end{equation}
We construct now $\MDBD$ $\BNTap$ as follows: for all $j\in[1:T]$ we set
\[
\alpha_{j}=\frac{w(p_{j})\cdot N}{S_{N+1,N}\cdot(T+1)},\quad\text{where}\quad p_{j}=\frac{j}{T+1}.
\]
Moreover, $\BNTap$ induces a vector $\gamma$ that satisfies $\gamma_{i}=\sum_{j=1}^{T}\alpha_{j}\cdot B_{N,i}(p_j)$ for all $i\in[0:N]$. By multiplying Equation \ref{eq:LAD} with $N/S_{N+1,N}$ we obtain
\begin{equation}\label{eq:approxGamma}
\left|\left(1+\eta_{i}\right)\hw(i/N)-\gamma_{i}\right|\leq \epsI/S_{N+1,N}.
\end{equation}
Furthermore, since
\[
\sum_{j=1}^{T}\alpha_{j} = \frac{N}{(T+1)\cdot S_{N+1,N}}\sum_{j=1}^{T}w(p_j)=\frac{S_{T,T+1}/(T+1)}{S_{N+1,N}/N},
\]
by Lemma \ref{lem_RatioLambda} there is a small positive number $\delta_w=o(1)$
such that
\[
\delta_w=1-\frac{S_{T,T+1}/(T+1)}{S_{N+1,N}/N}=1-\sum_{j=1}^{T}\alpha_{j}.
\]
Hence, we have
\begin{eqnarray*}
 &  & \sum_{i=0}^{N}\gamma_{i}=\sum_{i=0}^{N}\sum_{j=1}^{T}\alpha_{j}\cdot B_{N,i}(p_j) = \sum_{j=1}^{T}\alpha_{j}\sum_{i=0}^{N}B_{N,i}(p_j)=\sum_{j=1}^{T}\alpha_{j}=1-\delta_w.
\end{eqnarray*}
Since $\delta_w=o(1)$, by Equation~\ref{eq:approxGamma} it follows that $\gamma/[1-\delta_w]$ is a discretized probability distribution over $[0:N]$ that approximates component-wise the desired discretized p.d.f. $\hw$.

We note that the summations $S_{N+1,N}$ and $S_{T,T+1}$ can be computed in $O(N+T)$ work and $O(\log N+\log T)$ depth. Hence, the statement follows.
\end{proof}

\section{Efficient Parallel Solver for Transpose Bernstein-Vandermonde Systems}\label{sec:StBVls}

\begin{problem}\label{prob_ExN_Bnk}
Suppose a vector $\gamma\in(0,1)^{N+1}$ is
induced by a convex combination of exactly $N+1$ discrete Binomial
distributions $B(p_{i},N)$ such that $0<p_{i}\neq p_{j}<1$ for all $i\neq j$.
Find the unique vector $\alpha\in(0,1)^{N+1}$ such that $\gamma_{i}=\sum_{j=1}^{N+1}\alpha_{j} \cdot B_{N,i}(p_{j})$
for all $i\in[0:N]$.
\end{problem}

The Bernstein basis is a well studied primitive in the literature for polynomial interpolations~\cite{UN03,CUN02}. It is defined by $B_{N,k}(p)={N \choose k}p^{k}(1-p)^{N-k}$
for any $k\in[0:N]$. Let $\BNTap$ be $\MDBD$ with $T=N+1$. Then the Bernstein basis matrix is defined by $[\mB_{N}(p)]_{ji}=B_{N,i}(p_{j}),\,\,\forall i,j\in[1:N+1]$, and it has a full rank (c.f. Appendix \ref{appsec:BBM}). Moreover, the vector $\alpha$ is the unique solution of the linear system $\mB_{N}(p)^{\rot}\alpha=\gamma$.

In this section, we give an efficient parallel algorithm that solves Problem~\ref{prob_ExN_Bnk} and works in nearly linear work and poly-logarithmic depth. Our goal now is to prove Theorem~\ref{thm_my_Inv_Bnp}. We reduce a transpose Bernstein-Vandermonde system to a transpose Vandermonde system that can be solved efficiently and in parallel by a variation of an algorithm proposed by Gohberg and Olshevsky~\cite{GO94}. We present now their main algorithmic result.

\begin{thm}\cite{GO94}\label{thm_inv_VT}
There is an algorithm that on input two vectors
$\alpha,p\in\R^{N+1}$ such that $0<p_{i}\neq p_{j}<1$ for all
$i\neq j$, outputs the vector $\gamma=\Vp^{\rot}\alpha$ in $O(N\log^{2}N)$ time.
\end{thm}

Theorem~\ref{thm_inv_VT} follows by~\cite[Algorithm 2.1]{GO94} which relies on a non-trivial matrix decomposition of $\Vp^{\rot}$ to compute in $O(N\log^{2}N)$ time the desired  matrix-vector product. It can be easily verified that~\cite[Algorithm 2.1]{GO94} can be amended to compute the vector $\alpha=[\Vp^{\rot}]^{-1}\gamma$ in $O(N\log^{2}N)$ time. Furthermore, straightforward checking shows that this modified algorithm can be easily parallelized. We summarize below the resulting parallel algorithm.

\begin{thm}\cite{GO94}\label{thm_inv_pl_VT}
There is a parallel algorithm that on input two vectors $\gamma,p \in\R^{N+1}$ as in Theorem \ref{thm_inv_VT}, outputs the vector $\alpha=[\Vp^{\rot}]^{-1}\gamma$ in $O(N\log^{2}N)$ work and $O(\log^{c}n)$ depth, for some constant $c\in\N_+$.
\end{thm}

We prove now that the Bernstein basis matrix $\mB_{N}(p)$ admits the following decomposition.

\begin{lem}
\label{lem_BB_DVpD} Suppose $p\in(0,1)^{N+1}$ is vector
such that $0<p_{i}\neq p_{j}<1$ for all $i\neq j$, $\Vp$
is Vandermonde matrix defined by $[\Vp]_{ji}=(\frac{p_{j}}{1-p_{j}})^i$,
$D_{p}=\diag(\{ (1-p_{j})^{N}\} {}_{j=1}^{N+1})$
and $D_{CN}=\diag(\{ {N \choose i}\} _{i=0}^{N})$
are positive diagonal matrices. Then it holds that $\mBNp =D_{p}\cdot\Vp\cdot D_{CN}$.\end{lem}
\begin{proof}
By definition $[D_{p}\cdot\Vp\cdot D_{CN}]_{j,i} = (1-p_{j})^{N}(\frac{p_{j}}{1-p_{j}})^{i}{N \choose i} = B_{N,i}(p_{j}) = [\mB_{N}(p)]_{ji} $.
\end{proof}

We are ready now to prove Theorem \ref{thm_my_Inv_Bnp}.

\begin{proof}[Proof of Theorem~\ref{thm_my_Inv_Bnp}] By Lemma \ref{lem_full_Rank} the Bernstein matrix $\mBNp$ is invertible. Given a vector
$\gamma\in(0,1)^{N+1}$ we want to find the vector $\alpha=[\mBNp^{\rot}]^{-1}\gamma$.
By Lemma \ref{lem_BB_DVpD} we have $[\mBNp^{\rot}]^{-1} = [D_{p}]^{-1}\cdot[\Vp^{\rot}]^{-1}\cdot[D_{CN}]^{-1}$. Moreover, we can compute a vector $\gamma^{\prm}=[D_{CN}]^{-1}\gamma$
in $O(n)$ time. Using Theorem \ref{thm_inv_pl_VT}, we obtain a vector
 $\gamma^{\prm\prm}=[\Vp^{\rot}]^{-1}\gamma^{\prm}$
in $O(N\log^{2}N)$ work and $O(\log^{c}n)$ depth. The desired vector $\alpha=[D_{p}]^{-1}\gamma^{\prm\prm}$
takes further $O(n)$ work and $O(1)$ depth to compute.
\end{proof}

\subparagraph*{Acknowledgements}

We are grateful to Arijit Ghosh and Kunal Dutta for the helpful discussions on mixture of Binomial distributions, and to Arnur Nigmetov and Shay Moran for pointing us to the Bernstein interpolation polynomials. We would also like to thank Kurt Mehlhorn for the insightful comments and suggestions to the early version of the manuscript.

This work has been funded by the Cluster of Excellence ``Multimodal Computing and Interaction'' within the Excellence Initiative of the German Federal Government.



\bibliography{reference}




\appendix

\section{Generalized Escaping Probability}\label{appsec:GenEscProb}

\begin{proof}[Proof of Lemma~\ref{lem_mulApproxEGEP}]
By definition, for every vector $x$ the spectral sparsifier $D-\hA$ preserves approximately the quadratic form
\[
x^{\rot}(D-\hA)x\in[(1\pm\eps)\cdot x^{\rot}(D-D\mGg)x].
\]
Hence, the statement follows by applying the identities
\begin{eqnarray*}
\xi_{S}^{\rot}(D-D\mGg)\xi_{S} & = & \pi_{S}^{\rot}(I-\mGg)\mat{1}_{S} = 1-\pi_{S}^{\rot}\mGg\mat{1}_{S} = \pi_{S}^{\rot}\mGg\mat{1}_{\overline{S}}\\
 & = & \mathbb{E}_{v\sim\pi_{S}}\left[\mathrm{gEsc}(v,S,\mGg)\right].
\end{eqnarray*}
\end{proof}

\section{Spectral Sparsification of $\GL$-Matrices}\label{appsec:GLM}

Our proof of Lemma \ref{lem_SDDM_Closure} is based on the Perron-Frobenius Theorem~\cite{M73} for non-negative matrices.

\begin{thm}\label{thm_PF_NNM}\cite[Perron-Frobenius]{M73}
Suppose $A$ is symmetric nonnegative matrix. Then it has a nonnegative eigenvalue $\lambda$ which is greater than or equal to the modulus of all other eigenvalues.
\end{thm}

\begin{lem}\label{lem_SDDM_Closure}
Suppose $D-M$ is $\SDDM$ matrix. Then $D-D(\Di M)^{N}$ is $\SDDM$ matrix $\forall N\in\N_{+}$.
\end{lem}

\begin{proof}
Since $D-M$ is $\SDDM$ we have $D\succ M$. By Fact \ref{fact_fiveProps}.e it holds $I\succ \Dhm M \Dhm\triangleq X$. Hence, the largest eigenvalue $\lambda(X)<1$. By Theorem \ref{thm_PF_NNM} the spectral radius $\rho(X)<1$, i.e. $|\lambda_i(X)|<1$ for all $i$. Since $X$ is symmetric it has the form $X=\sum_{i}\lambda_{i}u_{i}u_{i}^{\rot}$. Moreover, we have $X^k=\sum_{i}\lambda_{i}^{k}u_{i}u_{i}^{\rot}$ for every $k\in\N_{+}$. Thus the spectral radius of $X^k$ satisfies $\rho(X^{k})=\rho(X)^{k}<1$.

Notice that $B_k = D-D(\Di M)^{k}$ is symmetric non-negative matrix for every $k\in\N_{+}$. By definition $D-M$ is diagonally dominant and thus $\Di M\mat{1}\preceq\mat{1}$ component-wise. This implies that $B_k$ is diagonally dominant matrix. Notice that $B_k = \Dhp[ I - X^{k} ]\Dhp$ and since $\rho(X^{k})<1$ it follows that $B_k$ is positive definite and hence $\SDDM$ matrix.
\end{proof}

\paragraph*{Proof of Lemma~\ref{lem_Closure}} We combine Lemma~\ref{lem_SDDM_Closure} with the following two statements.

\begin{fact}\label{fact_Lap_SDDM}
Spielman and Peng~\cite[Proposition 5.6]{PS14} showed that if $D-M$ is $\SDDM$ matrix, then $D-M\Di M$ is $\SDDM$ matrix. Also Cheng et al.~\cite[Proposition 25]{CCLPT15} showed that if $D-M$ is Laplacian matrix then $D-D(\Di M)^N$ is Laplacian matrix for every $N\in\N_{+}$.
\end{fact}

Based on Lemma~\ref{lem_SS_GL} and Fact~\ref{fact_Lap_SDDM} we establish the following result.

\begin{lem}\label{lem_Lap_SDDM_Sparsifiers}
Suppose $D-M$ is $\GL$-matrix and $D-\hM\approx_{\eps}D-M$ is a spectral sparsifier. Then $D-D(\Di \hM)^{N}$ is $\GL$-matrix for every $N\in\N_{+}$.
\end{lem}

\paragraph*{Proof of Lemma~\ref{lem_my_DappToD}} We use the following result that appears in Peng's thesis~\cite{PengPhd13}.

\begin{lem}
\cite{PengPhd13}\label{lem_DA}
Suppose $D-A$ is Laplacian matrix
$(\text{possibly }A_{ii}\neq0)$, and $\wD-\wA$
a sparsifier with $\wA_{ii}=0$ for every $i$ such that
$(1-\eps)(D-A)\preceq\wD-\wA\preceq D-A$.
Then the symmetric non-negative matrix $\hA=(D-\wD)+\wA$
satisfies $(1-\eps)(D-A)\preceq D-\hA\preceq(D-A)$.
\end{lem}

\begin{proof}[Proof of Lemma~\ref{lem_my_DappToD}]
Let $\widetilde{D_{1}}=\frac{1}{1+\eps}\wD$
and $\widetilde{A_{1}}=\frac{1}{1+\eps}\wA$. Then $\frac{1-\eps}{1+\eps}(D-A)\preceq\widetilde{D_{1}}-\widetilde{A_{1}}\preceq D-A$
and by Lemma \ref{lem_DA} the symmetric non-negative matrix $\hA=(D-\widetilde{D_{1}})+\widetilde{A_{1}}$
satisfies $\frac{1-\eps}{1+\eps}(D-A)\preceq D-\hA\preceq D-A.$
Since $\frac{1-\eps}{1+\eps}\geq1-2\eps$ for every $\eps\in(0,\frac{1}{2})$
the statement follows.
\end{proof}

\subsection{Structural Result}\label{appsubsec:SR}

Suppose $D-M$ is $\GL$-matrix. We show that the matrix $D-M\Di M$ can be expressed as a sum of a non-negative main diagonal matrix and a sum of Laplacian matrices.

\begin{proof}[Proof of \lemref{lemGLStruct}]
Let $M\in\mathbb{R}^{n\times n}$ and $nnz(M)=m$. We decompose the entries of matrix $M\Di M$ into three types. We set type $1$ to be the entries $(M\Di M)_{ii}=\sum_{k=1}^{n}M_{ik}^{2}/D_{k}$ for all $i$. We note that all entries of type $1$ can be computed in $O(m)$ time. We consider next the off-diagonal entries
\[
(M\Di M)_{ij}=\begin{array}{c}
\underbrace{(M_{ii}/d_{i}+M_{jj}/d_{j})\cdot M_{ij}}\\
\text{type 2}
\end{array}\begin{array}{c}
+\\
\\
\end{array}\begin{array}{c}
\underbrace{\sum_{k\neq\{ i,j\} }^{n}M_{ik}M_{jk}/D_{kk}}\\
\text{type 3}
\end{array}.
\]
Observe that the number of type $2$ entries is at most $m$. Now for a fixed $k$ we note that the corresponding entries that appear in type $1$ and type $3$ form
a weighted clique (with self-loops) whose adjacency matrix is defined by $\frac{1}{d_{k}}\eta_{k}\eta_{k}^{\rot}$.

Straightforward checking shows that $M\Di M=\mB+\sum_{i}\frac{1}{d_{i}}\eta_{i}\eta_{i}^{\rot}$.
By Lemma~\ref{lem_Closure} $D-M\Di M$ is $\GL$-matrix and thus diagonally dominant. Hence, the Laplacian matrices $\mL_{B}$ and $\mL_{N_i}$ for all $i$ exist. Moreover, we can compute in $O(m)$ time the positive diagonal matrices $\mD_{B}$ and $\mD_{N_i}=(s_{i}/d_{i})\diag(N_i)$ for all $i$. To see this, observe that $s_i$ and $d_i$ can be computed in $O(m)$ time for all $i$, and the number of elements in the disjoint union $|\sqcup_{i}N_i|\leq m$.
\end{proof}

\section{Bernstein Basis Matrix}\label{appsec:BBM}

We prove below that the Bernstein basis matrix in Problem \ref{prob_ExN_Bnk} has full rank.

\begin{lem}
\label{lem_full_Rank}
Suppose a vector $p\in(0,1)^{N+1}$ satisfies
$0<p_{i}\neq p_{j}<1$ for all $i\neq j$. Then the Bernstein
basis matrix $\mB_{N}(p)$ has a full rank.
\end{lem}

\begin{proof}
Suppose for contradiction that $\mathrm{rank}(\mB_{N}(p))<N+1$.
Then there is a vector $\lambda\in\mathbb{R}^{N+1}$ such that the
linear combination of the columns of $\mB_{N}(p)$
satisfies $\sum_{i=0}^{N}\lambda_{j}[\mB_{N}(p)]_{:,i}=0$.
Let $f_{\lambda}(x)$ be a polynomial defined by $f_{\lambda}(x)\triangleq\sum_{i=0}^{N}\lambda_{i}\cdot B_{N,i}(x)=\sum_{i=0}^{N}\lambda_{i}\cdot{N \choose i}x^{i}(1-x)^{N-i}$.
Notice that $f_{\lambda}(p_{j})=0$ for every $j\in[1:N+1]$,
i.e. $f_{\lambda}(x)$ has $N+1$ roots. However, since
the polynomial $f_{\lambda}(x)$ has degree $N$ it follows
that $f_{\lambda}(x)\equiv0$. Therefore, we obtained the desired
contradiction.
\end{proof}

\section{Approximating Two Canonical PDFs}\label{appsec:AGPD}

Here, we illustrate the representational power of $\MDBD$. We show that there are $\MDBD$ satisfying the hypothesis in Theorem~\ref{thm_AppDscrPDF} and approximate two canonical continuous p.d.f.: the Uniform distribution and the Exponential Families. More precisely, we prove that the Uniform distribution and the Exponential families admit a multiplicative and an additive approximation, respectively.

\subsection{Uniform Distribution}

\begin{lem}[Uniform Distribution]\label{lem_Approx_Bni}
Let $w(x)=1$, $N\in\N_{+}$ and $\eps>0$. If $T\geq\Omega(N\eps^{-1/2})$ then it holds that
\[
\frac{1}{T+1}\sum_{j=1}^{T}B_{N,i}\left(\frac{j}{T+1}\right)\in\left[(1\pm\eps)\frac{1}{N+1}\right], \text{ for all}\,\,i\in[3:N-3].
\]
\end{lem}

\begin{proof}
By Theorem \ref{thm_approx_Def_Intergral} we have that
\[
\left|\sum_{i=1}^{T}L_{i}\right|\leq\frac{1}{2}\sum_{i=1}^{T}\int_{x_{i-1}}^{x_{i}}\left|\frac{1}{4N^{2}}-(t-c_{i})^{2}\right|\left|\frac{d^{2}B_{N,i}(x)}{dx^{2}}(t)\right|dt\leq\frac{1}{8N^{2}}\int_{0}^{1}\left|\frac{d^{2}B_{N,i}(x)}{dx^{2}}(t)\right|dt.
\]
By combining Lemma~\ref{lem_int_n1} and Corollary \ref{cor_integral_derBB} for every $i\in[3:N-3]$
it holds
\[
\left|\frac{1}{N+1}-\frac{1}{T+1}\sum_{j=1}^{T}B_{N,i}\left(\frac{j}{T+1}\right)\right|\leq\left|\sum_{i=1}^{T}L_{i}\right|\leq\frac{1}{8T^{2}}\int_{0}^{1}\left|\frac{d^{2}B_{N,i}(x)}{dx^{2}}(x)\right|dx\leq\frac{N}{2T^{2}}.
\]
We note that $\frac{N}{2T^{2}}\leq\frac{\eps}{N+1}$ since $T\geq\Omega(N\eps^{-1/2})$, and hence the statement follows.
\end{proof}

\begin{rem}
All conditions of Theorem~\ref{thm_AppDscrPDF} hold. Note that $w(x)=1\cdot 1$, i.e., $f(x)=1$.\\
a) $C=1$, b) $f(x)=1$, c) $\frac{1}{2}[f(0)+f(1)]=1$ and d) $\int_{0}^{1}|f^{(2)}(x)|dx=0$, since $f^{(2)}(x)=0$.
\end{rem}

\subsection{Exponential Families}

\begin{lem}[Exponential Families]\label{lem_approxExpFam}
Let $N\in\N$, $k\in[1,\sqrt{N}]$ and $w(x)=\frac{k}{1-e^{-k}}\cdot\exp\{ -k\cdot x\}$
is a probability density function. For any $\eps>0$ if $T\geq\Omega(N\sqrt{k/\eps})$ then for every $i\in[3:N-3]$ it holds for the function
$F_{i}(x)=w(x)\cdot B_{N,i}(x)$ that
\[
\left|(1+\eta_{i})\frac{w(i/N)}{N}-\frac{1}{T+1}\sum_{j=1}^{T}F_{i}\left(\frac{j}{T+1}\right)\right|\leq\frac{\eps}{N},\quad\text{where}\quad\left|\eta_{i}\right|\leq\frac{1}{4}.
\]
\end{lem}

\begin{proof}
The $p$th derivative of function $w(x)$ satisfies $w^{(p)}(x)=(-k)^{p}\cdot w(x)$. Let $I=[0,1]$ be an interval. Straightforward checking shows that
\begin{equation}\label{eq:pthDerW}
\max_{x\in I}|w^{(p)}(x)| = k^p \cdot \max_{x\in I}|w(x)| < 2\cdot k^{p+1}.
\end{equation}
By the definition of function $b_1(x)$ (c.f. Theorem~\ref{thm_Hald68}) we have
\[
b_{1}(x)=-\frac{k^2}{2}\cdot x^{2} + \left(2k + \frac{k^2}{2} \right)\cdot x -(1+k),
\]
and we can show that $\max_{x\in I}b_{1}(x)\leq 1+k^2/8 \leq 1+N/8$. The function $b_2(x)$ satisfies
\[
b_{2}(x)=1+3(1-2x)k+(1-6x+6x^{2})k^2-\frac{5}{6}x(1-x)(1-2x)k^3 + \frac{1}{8}x^{2}(1-x)^{2}k^{4},
\]
and we can show that $\max_{x\in I}b_{2}(x)\ll k^{4}/8\leq N^2/8$. By Theorem~\ref{thmMDBD} it suffices to upper bound the following four cases.\\
\\
\textbf{Case 1:} By Theorem \ref{thm_Hald68} $\mu\leq\frac{1}{4}$, since $\max_{x\in I}|b_{1}(x)|\leq 1+N/8$
and $\max_{x\in I}|b_{2}(x)|\ll N^{2}/8$.\\
\\
\textbf{Case 2:} We combine Equation \ref{eq:pthDerW}
and Lemma \ref{lem_int_n1} to obtain
\[
\int_{0}^{1}|w^{\prm\prm}(t)\cdot B_{N,i}(t)|dt < 2\cdot k^{3}\cdot\int_{0}^{1}|B_{N,i}(t)|dt=\frac{2\cdot k^{3}}{N+1}.
\]
\textbf{Case 3:} By combining Equation \ref{eq:pthDerW} and Lemma \ref{cor_integral_derBB} it holds
\[
\int_{0}^{1}|w^{\prm}(t)\cdot B_{N,i}^{\prm}(t)|dt < 2\cdot k^{2}\cdot\int_{0}^{1}|B_{N,i}^{\prm}(t)|dt < 4\cdot k^{2}.
\]
\textbf{Case 4:} We use again Equation \ref{eq:pthDerW} and Lemma \ref{cor_integral_derBB} to obtain
\[
\int_{0}^{1}|w(t)\cdot B_{N,i}^{\prm\prm}(t)|dt < 2\cdot k \cdot \int_{0}^{1}|B_{N,i}^{\prm\prm}(t)|dt
< 8\cdot k\cdot N.
\]
Recall that $F_i(x)=w(x)B_{N,i}(x)$. By combining $|\frac{d^{2}F_i(x)}{dx^{2}}|=w^{\prm\prm}\cdot B_{N,i}+2w^{\prm}\cdot B_{N,i}^{\prm}+w\cdot B_{N,i}^{\prm\prm}$ and Theorem \ref{thm_approx_Def_Intergral} we have
\[
\left|\sum_{i=1}^{T}L_{i}\right| \leq \frac{1}{8T^{2}}\int_{0}^{1}\left|\frac{d^{2}F(x)}{dx^{2}}(t)\right|dt \lesssim \frac{k\cdot N}{T^2}.
\]
Hence, the statement follows.
\end{proof}

\begin{rem}
The hypothesis in Theorem~\ref{thm_AppDscrPDF} holds. Note that $w(x)=\frac{k}{1-e^{-k}}\cdot\exp\{ -k\cdot x\}$ and $\phiw=\frac{k}{1-e^{-k}}$. Thus  $C=\frac{k}{1-e^{-k}}$, $f(x)=\exp\{ -k\cdot x\}$, $f^{(2)}(x)=k^2\cdot f(x)$ and $k\in[1,\sqrt{N}]$. Furthermore,\\
a) $C=o(N)$, b) $0\leq f(x)\leq1$, c) $\frac{1}{2}\leq\frac{1}{2}[f(0)+f(1)]<1$ and for d) we have
\[
\int_{0}^{1}|f^{(2)}(x)|dx = k^{2}\int_{0}^{1}|f(x)|dx = \frac{k^{2}}{C}=(1-e^{-k})k = o(N).
\]
\end{rem}

\section{Schur Complement}\label{appsec:ASC}

In this section we prove \lemref{lemSchurRec}. We use the following result proposed by Peng et al.~\cite{arxivCCLPT15}.

\begin{lem}
\label{lem_DPM}\cite[Lemma 4.3]{arxivCCLPT15} If $M$ is \emph{$\SPSD$}
matrix and $(1-\eps)(D-M)\preceq D-\hM \preceq(1+\eps)(D-M)$
then it holds that $(1-\eps)(D+M)\preceq D+\hM \preceq(1+\eps)(D+M).$
\end{lem}

We extend next two technical results on Schur complement that appeared in~\cite{PengPhd13,MP13,CCLPT15}.

\begin{clm}\label{clm_Schur}
Suppose $X\triangleq\left[\begin{array}{cc}
P_{1} & -M\\
-M & P_{2}
\end{array}\right]$ where $P_{1}$ and $P_{2}$ are symmetric positive definite matrices
and $M$ is symmetric matrix. Then $\,$ $v^{\rot}[P_{2}-MP_{1}^{-1}M]v=\min_{u}\left(\begin{array}{c}
u\\
v
\end{array}\right)^{\rot}X\left(\begin{array}{c}
u\\
v
\end{array}\right)$ for every $v$.
\end{clm}

\begin{proof}
Suppose $f_{v}(u)\triangleq\left(\begin{array}{c}
u\\
v
\end{array}\right)^{\rot}\left[\begin{array}{cc}
P_{1} & -M\\
-M & P_{2}
\end{array}\right]\left(\begin{array}{c}
u\\
v
\end{array}\right)=u^{\rot}P_{1}u-2u^{\rot}Mv+v^{\rot}P_{2}v$. Notice that $f$ is minimized when $u=P_{1}^{-1}Mv$, since $\nabla f_{v}(u)=2P_{1}u-2Mv$.
Hence, it follows that $\min_{u}\left(\begin{array}{c}
u\\
v
\end{array}\right)^{\rot}X\left(\begin{array}{c}
u\\
v
\end{array}\right)=v^{\rot}P_{2}v-v^{\rot}MP_{1}^{-1}Mv=v^{\rot}[P_{2}-MP_{1}^{-1}M]v.$
\end{proof}

\begin{lem}
\label{lem_Schur}(Schur Complement) Suppose $D_{1},D_{2}$ are positive main diagonal matrices
and $M,Q$ are symmetric matrices. If $X_{M}\triangleq\left[\begin{array}{cc}
D_{1} & -M\\
-M & D_{2}
\end{array}\right]\approx_{\eps}\left[\begin{array}{cc}
D_{1} & -Q\\
-Q & D_{2}
\end{array}\right]\triangleq X_{Q}$ then it holds that $D_{2}-MD_{1}^{-1}M\approx_{\eps}D_{2}-QD_{1}^{-1}Q$.
\end{lem}

\begin{proof}
Here we show the upper bound, but
the lower bound follows by analogy. Let $w$ be a vector such that
$\left(\begin{array}{c}
w\\
v
\end{array}\right)^{\rot}X_{Q}\left(\begin{array}{c}
w\\
v
\end{array}\right)=\min_{u}\left(\begin{array}{c}
u\\
v
\end{array}\right)^{\rot}X_{Q}\left(\begin{array}{c}
u\\
v
\end{array}\right)$. We apply twice Claim \ref{clm_Schur}
to obtain the following chain of inequalities
\begin{eqnarray*}
v^{\rot}[D_{2}-MD_{1}^{-1}M]v&=&\min_{u}\left(\begin{array}{c}
u\\
v
\end{array}\right)^{\rot}X_{M}\left(\begin{array}{c}
u\\
v
\end{array}\right)\leq\left(\begin{array}{c}
w\\
v
\end{array}\right)^{\rot}X_{M}\left(\begin{array}{c}
w\\
v
\end{array}\right)\\&\leq&(1+\eps)\left(\begin{array}{c}
w\\
v
\end{array}\right)^{\rot}X_{Q}\left(\begin{array}{c}
w\\
v
\end{array}\right)=(1+\eps)\min_{u}\left(\begin{array}{c}
u\\
v
\end{array}\right)^{\rot}X_{Q}\left(\begin{array}{c}
u\\
v
\end{array}\right)\\&=&(1+\eps)v^{\rot}[D_{2}-QD_{1}^{-1}Q]v.
\end{eqnarray*}
\end{proof}

We prove \lemref{lemSchurRec} by arguing in a similar manner to in~\cite[Lemma 4.4]{arxivCCLPT15}. We present the proof here for completeness.

\begin{proof}[Proof of \lemref{lemSchurRec}] For any symmetric matrix $X$, we denote
by $\mathcal{P}_{X}=\left[\begin{array}{cc}
D & -X\\
-X & D
\end{array}\right]$. We prove the upper bound, but the lower bound follows by analogy.
Straightforward checking shows that
\begin{eqnarray*}
 &  & \left(\begin{array}{c}
u\\
v
\end{array}\right)^{\rot}\mathcal{P}_{\hM }\left(\begin{array}{c}
u\\
v
\end{array}\right)=u^{\rot}Du-v^{\rot}\hM u-u^{\rot}\hM v+v^{\rot}Dv\\
 & = & \frac{1}{2}\left[(u+v)^{\rot}(D-\hM )(u+v)+(u-v)^{\rot}(D+\hM )(u-v)\right]
\end{eqnarray*}
By Lemma \ref{lem_DPM} it holds that $D+\hM \approx_{\eps}D+M$
and thus we have
\begin{eqnarray*}
\left(\begin{array}{c}
u\\
v
\end{array}\right)^{\rot}\mathcal{P}_{\hM}\left(\begin{array}{c}
u\\
v
\end{array}\right)&=&\frac{1}{2}\left[(u+v)^{\rot}(D-\hM)(u+v)+(u-v)^{\rot}(D+\hM)(u-v)\right]\\&\leq&(1+\eps)\frac{1}{2}\left[(u+v)^{\rot}(D-M)(u+v)+(u-v)^{\rot}(D_{2}+M)(u-v)\right]\\&=&(1+\eps)\left(\begin{array}{c}
u\\
v
\end{array}\right)^{\rot}\mathcal{P}_{M}\left(\begin{array}{c}
u\\
v
\end{array}\right).
\end{eqnarray*}
Hence, it follows that $\mathcal{P}_{\hM }=\left[\begin{array}{cc}
D & -\hM \\
-\hM  & D
\end{array}\right]\approx_{\eps}\left[\begin{array}{cc}
D & -M\\
-M & D
\end{array}\right]=\mathcal{P}_{M}$. Now, by Lemma \ref{lem_Schur} it holds that $D-\hM \Di \hM \approx_{\eps}D-M\Di M$.
\end{proof}

\end{document}